\documentclass[preprint,pre,twocolumn,10pt]{revtex4}%
\usepackage{amsfonts}
\usepackage{amsmath}
\usepackage{amssymb}
\usepackage{graphicx}%
\providecommand{\U}[1]{\protect\rule{.1in}{.1in}}
\newtheorem{theorem}{Theorem}

\newtheorem{definition}[theorem]{Definition}

\newtheorem{proposition}[theorem]{Proposition}

\newenvironment{proof}[1][Proof]{\noindent\textbf{#1.} }{\ \rule{0.5em}{0.5em}}
\begin{document}
\preprint{UATP/1801}
\title{Hierarchy of Relaxation times and Residual Entropy: A Nonequilibrium Approach}
\author{P. D. Gujrati}
\email{pdg@uakron.edu}
\affiliation{Department of Physics, Department of Polymer Science, The University of Akron,
Akron, OH 44325}

\begin{abstract}
We consider nonequilibrium (NEQ) states such as supercooled liquids and
glasses that are described with use of internal variables. We classify the
latter by state-dependent hierarchy of relaxation times to assess their
relevance for irreversible contributions. Given an observation time
$\tau_{\text{obs}}$, we determine the window of relaxation times that divide
the internal variables into active and inactive groups, the former playing a
central role in the NEQ thermodynamics. Using this thermodynamics, we
determine (i) a bound on the NEQ entropy and on the residual entropy, and (ii)
the nature of isothermal relaxation of the entropy and the enthalpy in
accordance with the second law. A theory that violates the second law such as
the entropy loss view is shown to be internally inconsistent if we require it
to be consistent with experiments. The inactive internal variables still play
an indirect role in determining the temperature $T(t)$, the pressure $P(t)$,
of the system, which deviate from their external values.

\end{abstract}
\date[January 24, 2018]{}
\maketitle

\section{Introduction}

Glass such as naturally occurring obsidian, pumice, etc. or man-made Venetian
glass, window glass, etc. is a well-known class of material that has captured
our fascination forever. We can now make a defect-free glass in the laboratory
for a variety of scientific and technological applications. Crudely speaking,
it is an almost solid-like amorphous material that possesses no long range
atomic order and, upon heating, \emph{gradually} softens as it turns into its
molten state (also known as the supercooled liquid) as it passes through the
glass transition region normally denoted by a suitable chosen single
temperature $T_{\text{g}}$ in this region
\cite{Goldstein-Ann,Nemilov-Book,Debenedetti,Gutzow-Book}. For the purpose of
this article, a glass is treated merely as a \emph{nonequilibrium} (NEQ) state
of matter, which can be made quite homogeneous so to a good approximation it
can be treated as a thermodynamic system that is in \emph{internal
equilibrium} (IEQ) but not in \emph{equilibrium} (EQ)\ as explained later. (At
present, it suffices to say that the entropy in an IEQ state is a state
function of its state variables that now include some NEQ state variables
(commonly known as internal variables)
\cite{Goldstein-Ann,Nemilov-Book,Debenedetti,Gutzow-Book} besides those needed
to specify EQ states; see also \cite{Note-IEQ,Simon,Landau}.) This means that
a glass will exhibit relaxation as it strives to come to equilibrium. The
relaxation time is known to be large enough close to $T_{\text{g}}$ that at
much lower temperatures, one can usually treat a glass to be in a almost
\emph{frozen} state over experimental time scale $\tau_{\text{obs}}$, the time
period over which successive observations are made. We refer the reader to an
excellent monograph by Debenedetti \cite{Debenedetti} on these issues. We will
primarily focus on the thermodynamics of glasses and supercooled liquid in
this work and treat them as NEQ states. Therefore, our discussion will mostly
consider a NEQ system, which we denote by $\Sigma$ in an extensively large
medium $\widetilde{\Sigma}$ as shown in Fig. \ref{Fig.System}.

\begin{definition}
As we will not consider a system in isolation in this work, we will always use
EQ\ or "equilibrium" to mean "equilibrium with respect to the medium
$\widetilde{\Sigma}$." We will not reserve EQ for the entire system only. We
will also use it for a part of the system, part of the state variables, or
part of the degrees of freedom such as vibrational degrees, of the system, if
they are in equilibrium with $\widetilde{\Sigma}$. On the other hand, we will
reserve the use of IEQ for the entire system; see also \cite{Note-IEQ}.
\end{definition}

It is a well-known fact that in glasses, the vibrational modes come to
equilibrium very fast, even though the glass is out of equilibrium. Similarly,
in a sinusoidal variation of $T$, some degrees of freedom would equilibrate
after a cycle; others would not and would control the temporal behavior of the
system. It seems natural that the sinusoidal variation would give rise to a
distribution of relaxation times. Thus, in general, one of the most important
consequences of the rate of variation of the external stimuli such as the
temperature or pressure is the possibility that the state of the system may be
so far away from equilibrium that the dynamics becomes too complex, involving
multiple relaxation time scales $\tau_{0},\tau_{1},\tau_{2},\cdots$, in
supercooled liquids \cite{Gotze,Cummins,Debenedetti,Goldstein-Ann}. The
relaxation time is defined as the time required for the corresponding
dynamical variable to come to equilibrium with the medium; see Eq.
(\ref{Relaxation-Def}) for the proper definition of the relaxation time. It
should be emphasized that this interpretation of the relaxation time is
dictated by the experimental setup but does not depend on any particular
mathematical form of the relaxation. An interplay between $\tau_{\text{obs}}$
and relaxation times $\tau_{k}$'s becomes crucial in determining the
thermodynamics of the system and plays a major role in our discussion here. In
fact,\ one of the following cases for a given $\tau_{k}$ will be usually
encountered in experiments:

\begin{enumerate}
\item[Relax1] $\tau_{k}<<\tau_{\text{obs}}$. In this situation, the $k$th
relaxing dynamical variable has equilibrated and does not have to be accounted
for in the NEQ thermodynamics.

\item[Relax2] $\tau_{k}\simeq\tau_{\text{obs}}$. In this situation, the $k$th
dynamical variable will continue to relax towards equilibrium during
$\tau_{\text{obs}}$ and must be accounted for as the system approaches equilibrium.

\item[Relax3] $\tau_{k}>>\tau_{\text{obs}}$. In this situation, the $k$th
dynamical variable will not fully relax and will strongly affect the behavior
of the system. The corresponding dynamical variable is said to be "frozen-in"
over $\tau_{\text{obs}}$.
\end{enumerate}

When there are several relaxation times, it is possible that different
$\tau_{k}$'s will correspond to different cases above. Thus, care must be
exercised in dealing with different relaxation times. The need for \ such a
care has been recognized in vitrification for a long time \cite{Wilks}.
Relaxation is a universal phenomenon when a system drives itself towards a
more stable state such as an EQ state. In liquids or glasses, relaxations
involving changes of the atomic or molecular positions are generally known as
\emph{structural relaxations} \cite{Scherer}. Recent experimentation advances
have made it possible to directly measure these relaxation processes at the
molecular level simultaneously \cite{Ediger}. At sufficiently low
temperatures, the characteristic time for structural relaxations becomes
comparable to the time scale of a macroscopic observation $\tau_{\text{obs}%
}\sim$ $100$ s. For shorter time scales, the supercooled liquid (SCL) exhibits
solid-like properties, while for longer times, it shows liquid-like
properties. Even the dynamics in these cases is not so trivial but has been
investigated for a long time \cite{Gotze,Cummins,Debenedetti,Goldstein-Ann}
with tremendous success. The glass transition being a "NEQ transition," its
description will require extensive \emph{internal variables}, collectively
denoted by a vector $\boldsymbol{\xi}$ that are independent of the set
$\mathbf{X}$ of extensive \emph{observables} ($E,V,N,\cdots$) \cite{Note1}
whenever the system is out of equilibrium
\cite{Maugin,Gutzow-Book,Nemilov-Book,Gujrati-I,Gujrati-II,deGroot,Prigogine,Beris,Kestin,Woods,Coleman,Jou0,Meixner}%
. We denote their collection by $\mathbf{Z}$ in this work. The investigations
of the glass transition invariably assume that the entropy $S$ is a
\emph{state function} $S(\mathbf{Z})$\ of the state variables in the
\emph{extended state space} $\mathfrak{S}_{\mathbf{Z}}\supset\mathfrak{S}%
_{\mathbf{X}}$ spanned by $\mathbf{Z}$; here $\mathfrak{S}_{\mathbf{X}}$ is
the state space of the observables. There is a memory of the initial state and
requires the entire history of how the state is prepared to uniquely describe
the preparation. Such a memory in some cases can be described by
$\boldsymbol{\xi}$. One example is residual stresses \cite{Withers}: if
particle configurations in a glass cannot fully relax to equilibrium, some of
the stresses that build up during flow in the melt persist in the glass; these
stresses cannot be captured by $\mathbf{X}$. We will say that such a state is
an incompletely described state in terms of $\mathbf{X}$ but a completely
described state in terms of $\mathbf{Z}$. In contrast, the EQ state
\textsf{M}$_{\text{eq}}(\mathbf{X})$ is a completely (\textit{i.e., }uniquely)
described state by $\mathbf{X}$ and has no memory of the initial state. This
means that in equilibrium, $\boldsymbol{\xi}$ is no longer independent of
$\mathbf{X}$.

The consideration of dynamics\ resulting from the simple connectivity of the
sample (also known as the microstate or phase) space has played a pivotal role
in developing the kinetic theory of gases \cite{Boltzmann0,Lebowitz}, where
the interest is at high temperatures
\cite{Landau,Gujrati-Residual,Gujrati-Symmetry,Gujrati-Poincare}. As dynamics
is very fast here, it is well known that the ensemble averages agree with
temporal averages. However, at low temperatures, where dynamics becomes
sluggish as in a glass \cite{Gujrati-book,Palmer,Debenedetti,Jackel}, the
system can be \emph{confined} into \emph{disjoint} components. The confinement
occurs under NEQ conditions, when the observational time scale $\tau
_{\text{obs}}$\ becomes shorter than the equilibration time $\tau_{\text{eq}}$
such as in glasses, whose behavior and properties have been extensively
studied. These components are commonly known as \emph{basins} in the energy
landscape picture \cite{GoldsteinLandscape,GujratiLandscape}. The
\emph{entropy of confinement} at absolute zero is known as the \emph{residual
entropy} and can be observed in glasses or disordered crystals; see below.%

\begin{figure}
[ptb]
\begin{center}
\includegraphics[
height=1.5229in,
width=3.1298in
]%
{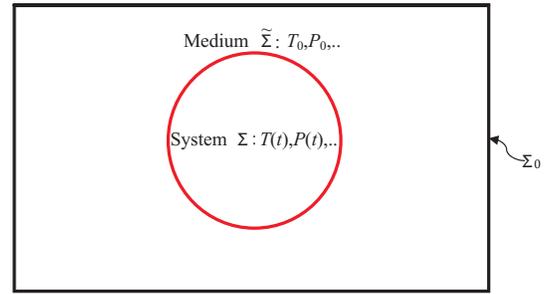}%
\caption{An isolated system $\Sigma_{0}$ consisting of the system $\Sigma$ in
a surrounding medium $\widetilde{\Sigma}$. The medium and the system are
characterized by their fields $T_{0},P_{0},...$ and $T(t),P(t),...$,
respectively, which are different when the two are out of equilibrium. }%
\label{Fig.System}%
\end{center}
\end{figure}

The existence of a nonzero residual entropy does not violate Nernst's
postulate, as the latter is applicable only to EQ states \cite[Sect.
64]{Landau}. The observation of residual entropy is very common in Nature.
Indeed, Tolman \cite[Sect. 137]{Tolman} devotes an entire section on this
issue for crystals in his seminal work, while Sethna provides an illuminating
discussion for glasses \cite[Sect. 5.2.2]{Sethna}. In addition, the existence
of the residual entropy has been demonstrated rigorously for glasses by Pauli
and Tolman \cite{Pauling} and for a very general spin model by Chow and Wu
\cite{Chow}; see references in these works for other cases where the residual
entropy is shown to exist rigorously. The numerical simulation carried out by
Bowles and Speedy for glassy dimers \cite{Speedy} also supports the existence
of a residual entropy. We refer the reader to consult various publications
\cite{Goldstein,Kozliak,Tolman}. Experiment evidence for a nonzero residual
entropy is abundant as discussed by several authors \cite[among others]%
{Gujrati-Residual,Giauque,Giauque-Gibson,Jackel,Pauling-ice,Nagle,Isakov,Berg,Speedy,Bestul}%
; various textbooks \cite{Gutzow-Book,Nemilov-Book} also discuss this issue.

We introduce useful notation and concepts in the next section. In the
following section, we introduce the concept of internal equilibrium (IEQ)
states for which the entropy is a state function in the extended state space
$\mathfrak{S}_{\mathbf{Z}}$.

\begin{definition}
As we are not interested in ordering phenomena (such as crystallization), we
define a NEQ state with respect to an EQ state that is also \emph{disordered},
\textit{i.e. }with respect to SCL.\textit{ This is formally done by
considering only disordered configurations and discarding all ordered
configurations in our discussion. We warn the reader that this is different
from the conventional approach in which the equilibrium state is always taken
to be the perfectly crystalline state. This point should not be forgotten. }We
then discuss the nature of the nonequilibrium state variables in
$\mathfrak{S}_{\mathbf{Z}}$ in Proposition 1. The affinity $\mathbf{A}$
corresponding to $\boldsymbol{\xi}$ is defined so that it vanishes in SCL, the
equilibrium state in our approach.
\end{definition}

The concept of a hierarchy of relaxation times is introduced in Sec.
\ref{Sec-Hierarchy}, which forms a central part of the paper. A given
$\tau_{\text{obs}}$ determines a particular time window, which provides a
justification of Proposition 1. We find that internal variable
$\boldsymbol{\xi}_{\text{E}}$ that has equilibrated play no role
thermodynamically since their affinity vanishes during $\tau_{\text{obs}}$. In
Sec. \ref{Sec_General_Consideration}, we discuss the first law in terms of the
new notation, identify the irreversible work, and the IEQ thermodynamics to be
used in the next two sections on the entropy bound in vitrification\ and the
residual entropy (Sec. \ref{Sec-Vitrification}) and on the properties of the
isothermal relaxation (\ref{Sect_Relaxation}). In Sec. \ref{Sec-Fictive}, we
find that $\boldsymbol{\xi}_{\text{E}}$ still indirectly affect thermodynamics
as it is required to have a thermodynamic temperature, pressure, etc. for the
system. The final section contains a brief discussion of the results.

\section{Notation\label{Sec-Notation}}

Below is a brief introduction to the notation and the significance of various
modern terminology \cite{deGroot,Prigogine} for readers who are unfamiliar
with them. As usual, $\Sigma$ and $\widetilde{\Sigma}$ form an isolated system
$\Sigma_{0}$. Extensive quantities associated with $\widetilde{\Sigma}$ and
$\Sigma_{0}$\ carry a tilde $\widetilde{\square}$ and a suffix $0$,
respectively. As $\widetilde{\Sigma}$ is very large compared to $\Sigma$ and
is in equilibrium, all its conjugate fields $T_{0},P_{0},$ etc. carry a suffix
$0$ as they are the same as for $\Sigma_{0}$, and there is no irreversibility
in $\widetilde{\Sigma}$. Any irreversibility is ascribed to the system
$\Sigma$ \cite{deGroot,Prigogine}, and is caused by processes such as
dissipation due to viscosity, internal inhomogeneities, etc. that are internal
to the system. Quantities without any suffix refer to the system. Throughout
this work, we will assume that $\Sigma$ and $\widetilde{\Sigma}$ are spatially
disjoint and \emph{statistically} quasi-independent
\cite{Gujrati-II,Gujrati-Entropy2,Gujrati-Entropy1} so that their volumes,
masses and entropies are additive at each instant. In particular,
$d\widetilde{V}=-dV$, since $V_{0}=V+\widetilde{V}$ remains constant for
$\Sigma_{0}$. We define a quantity to be system-intrinsic (SI) quantity if it
depends only on the property of the system alone and nothing else. For
example, if $P$ is the pressure of $\Sigma$ and $P_{0}$ that of $\widetilde
{\Sigma}$, then $PdV$ is the SI work done by the system, but $P_{0}dV$ is not
as the latter also depends on $\widetilde{\Sigma}$ through $P_{0}$. However,
$P_{0}d\widetilde{V}=-P_{0}dV$ is the work done by the medium, and this work
can be identified as a medium-intrinsic (MI) quantity. Any extensive SI
quantity $q(t)$ of $\Sigma$ can undergo two distinct kinds of changes in time:
one due to the exchange with the medium and another one due to internal
processes. Following modern notation \cite{Prigogine,deGroot}, exchanges of
$q(t)$ with the medium and changes within the system carry the suffix e and i,
respectively:%
\begin{equation}
dq(t)\doteq q(t+dt)-q(t)\equiv d_{\text{e}}q(t)+d_{\text{i}}q(t).
\label{q-partition}%
\end{equation}
For $\widetilde{\Sigma}$ and $\Sigma_{0}$, we must replace $q(t)$ by
$\widetilde{q}(t)$ and $q_{0}(t)$, respectively, so that $d\widetilde
{q}(t)=\widetilde{q}(t+dt)-\widetilde{q}(t)$ and $dq_{0}(t)\doteq
q_{0}(t+dt)-q_{0}(t)$. We will assume additivity so that
\[
q_{0}(t)=q(t)+\widetilde{q}(t).
\]
For this to hold, we need to assume that $\Sigma$ and $\widetilde{\Sigma}$
interact so \emph{weakly} that their interactions can be neglected. As there
is no irreversibility within $\widetilde{\Sigma}$ , we must have $d_{\text{i}%
}\widetilde{q}(t)=0$ for any medium quantity $\widetilde{q}(t)$ and
\begin{equation}
d_{\text{e}}q(t)\doteq-d\widetilde{q}(t)=-d_{\text{e}}\widetilde{q}(t).
\label{exchange-equality}%
\end{equation}
It follows from additivity that
\begin{equation}
dq_{0}(t)\equiv dq(t)+d\widetilde{q}(t)=d_{\text{i}}q(t).
\label{Isolated-system-irreversibility}%
\end{equation}
This means that any irreversibility in $\Sigma_{0}$ is ascribed to $\Sigma$,
and not to $\widetilde{\Sigma}$.\ In a reversible change, $d_{\text{i}%
}q(t)\equiv0$. For example, the entropy change
\[
dS\equiv d_{\text{e}}S+d_{\text{i}}S
\]
for $\Sigma$; here,
\[
d_{\text{e}}S=-d_{\text{e}}\widetilde{S}%
\]
is the entropy exchange with the medium and $d_{\text{i}}S$ is
\emph{irreversible entropy generation} due to internal processes within
$\Sigma$; the latter is also the entropy change $dS_{0}$\ of $\Sigma_{0}$; see
Eq. (\ref{Isolated-system-irreversibility}). Similarly, if $dW$ and $dQ$
represent the work done by and the heat change of the system, then
\begin{equation}
dW\equiv d_{\text{e}}W+d_{\text{i}}W,dQ\equiv d_{\text{e}}Q+d_{\text{i}}Q.
\label{Work-Heat-Partition}%
\end{equation}
Here, $d_{\text{e}}W$ and $d_{\text{e}}Q$ are the \emph{work exchange} and
\emph{heat exchange }with the medium, respectively, and $d_{\text{i}}W\ $and
$d_{\text{i}}Q$ are \emph{irreversible} work done and heat generation due to
internal processes in $\Sigma$. For an isolated system such as $\Sigma_{0}$,
the exchange quantity vanishes so that
\begin{equation}
dW_{0}(t)=d_{\text{i}}W_{0}(t);dQ_{0}(t)=d_{\text{i}}Q_{0}(t).
\label{Isolated-system-irreversibility-W-Q}%
\end{equation}

We have introduced the pressure-volume work. We identify $d_{\text{e}}%
W=P_{0}dV=-d_{\text{e}}\widetilde{W},dW=PdV$ and $d_{\text{i}}W=(P-P_{0})dV$.
In the absence of any chemical reaction, $dN_{k}=d_{\text{e}}N_{k}%
,d_{\text{i}}N_{k}=0$ for the $k$th species of the particles; otherwise,
$d_{\text{i}}N_{k}$ is its change due to chemical reaction within $\Sigma$. As
the energy of $\Sigma$ can only change due to exchange with $\widetilde
{\Sigma}$,
\begin{equation}
dE=d_{\text{e}}E,d_{\text{i}}E=0. \label{EnergyExchange}%
\end{equation}

We now explain the concept of the relaxation time used in this work, which is
a simple generalization of its common usage but which proves useful here.
Consider some dynamical variable $\Phi(t)$ as a function of time. Its
dependence on $\mathbf{Z}(t)$ is suppressed. Let $\Phi(\infty)$\ denote its
limiting value as $t\rightarrow\infty$; thus it also represents its EQ value.
In reality, we do not have to wait infinite amount of time as we cannot
distinguish between a nonzero difference $\left\vert \Phi(t)-\Phi
(\infty)\right\vert $, which is smaller than some small cutoff value so that
for all purposes it is no different than zero, or a zero difference. Let us
introduce a normalized ratio%
\[
\varphi(t)=\left\vert [\Phi(t)-\Phi(\infty)]/[\Phi(0)-\Phi(\infty
)]\right\vert
\]
to account for this cutoff value, which we denote by $e^{-\lambda}>0$; the
cutoff is primarily determined by the experimental setup. We say that the
dynamical variable $\Phi(t)$ has \emph{equilibrated }when $\varphi(t)$ equals
the cutoff $e^{-\lambda}$. The relaxation time $\tau_{\text{rel}}$ is defined
by%
\begin{equation}
\varphi(\tau_{\text{rel}})=e^{-\lambda}. \label{Relaxation-Def}%
\end{equation}
It is clear that for a given choice $\lambda$, the relaxation time
$\tau_{\text{rel}}$\ can be used to describe how rapidly a quantity
effectively reaches its equilibrium value. Usually, one assumes for
$\varphi(t)$ an exponential form
\[
\varphi(t)=\exp(-t/\tau)
\]
or a stretched exponential form
\[
\varphi(t)=\exp(-\left(  t/\tau\right)  ^{\beta}),0<\beta\leq1,
\]
also known as the Kohlrausch-Williams-Watts form, which reduces to the simple
exponential for $\beta=1$. The relaxation time is
\begin{equation}
\tau_{\text{rel}}=\lambda^{1/\beta}\tau, \label{relaxation time}%
\end{equation}
and reduces to $\tau_{\text{rel}}=\lambda\tau$ for $\beta=1$, the exponential
form. In this work, we do not make any particular choice for the decay
behavior of $\varphi(t)$; thus, we do not make any distinction between the two
forms of relaxation given above or any other form. We use a similar cutoff to
identify the equilibration time $\tau_{\text{eq}}$. In reality, the stretched
exponential is very common in glassy dynamics, but its origin is far from
clear at present, even though attempts have been made to express it as a
superposition of simple exponentials with different $\tau$'s
\cite{Montroll,Volchek}. It is, therefore, treated as empirical in nature. The
origin for the exponential relaxation, on the other hand, is well known as the
Debye dynamics. For us, what is important is the existence of $\tau
_{\text{rel}}$ through Eq. (\ref{Relaxation-Def}) and not the actual form of
$\varphi(t)$.

We find it very useful in this work to divide all internal variables in
$\boldsymbol{\xi}$ into \emph{nonoverlapping} groups $\boldsymbol{\xi}_{n}$
indexed by $n=1,2,\cdots$. All internal variables in $\boldsymbol{\xi}_{n}%
$\ are chosen to have the same relaxation time $\tau_{n}$ so that they
equilibrate and are no longer independent of $\mathbf{X}$ over time interval
$\Delta t\gtrsim\tau_{n}$, and that all groups have distinct relaxation times
($\tau_{i}\neq\tau_{j}$ for $i\neq j$). We supplement $\boldsymbol{\xi}$ by
introducing a new group $\boldsymbol{\xi}_{0}=\mathbf{X}$ with relaxation time
$\tau_{0}=\tau_{\text{eq}}$\ in order to compactify our notation so that
$\mathbf{Z}=\left\{  \boldsymbol{\xi}_{k}\right\}  _{k\geq0}$. We also
introduce the concept of hierarchy of relaxation times $\tau_{0}>\tau_{1}%
>\tau_{2}>\cdots$ associated with $\boldsymbol{\xi}_{0},\boldsymbol{\xi}%
_{1},\boldsymbol{\xi}_{2},\cdots$, and state spaces $\mathfrak{S}_{0}%
\subset\mathfrak{S}_{1}\subset\mathfrak{S}_{2}\subset\cdots$, where
$\mathfrak{S}_{n},n=0,1,2,\cdots$, is spanned by all $\boldsymbol{\xi}%
_{k},k\leq n$, with relaxation times $\tau_{k}>\tau_{n+1}$. Physically, the
hierarchy of relaxation times means that the longest relaxation time in
$\mathfrak{S}_{n}$ is $\tau_{0}$ corresponding to $\boldsymbol{\xi}%
_{0}=\mathbf{X}$\ and the shortest relaxation time is $\tau_{n}$ corresponding
to $\boldsymbol{\xi}_{n}$.\ Thus, if $\tau_{n}>\tau_{\text{obs}}$, any
$\boldsymbol{\xi}_{k},k>n$, with relaxation time shorter than $\tau_{n}$ has
already equilibrated (\textit{i.e.}, is no longer independent of $\mathbf{X}$)
and does not have to be used to specify the NEQ state. Thus, $\mathfrak{S}%
_{n}$ is the state space needed to specify the NEQ state for $\tau_{n}%
>\tau_{\text{obs}}$. However, as $T_{0}$ is changed, both $\left\{  \tau
_{n}\right\}  $ and $\tau_{\text{obs}}$ can change as shown in Fig.
\ref{Fig-Relaxation}. This then affects the choice of the required state space
$\mathfrak{S}_{n}$. Thus, the hierarchy becomes a central concept in our analysis.

One of the most important set of internal variables is that associated with
the vibrational modes in the system. We denote it by $\boldsymbol{\xi
}_{\text{v}}$ and seems to have the property that it is always inactive. This
is shown by the lowest lying relaxation time curve corresponding to
$\tau_{\text{v}}$ in Fig. \ref{Fig-Relaxation}. This is because, we expect
these modes to always come to equilibrium with the medium for any reasonable
$\tau_{\text{obs}}$.

\section{Generalized Nonequilibrium Thermodynamics in the Extended Space}

We are mostly interested in disordered states of a system in this work. Any
ordered state, if it exists, is taken out of the consideration from start.
Thus, the state space $\mathfrak{S}_{\mathbf{X}}$ only contains disordered
states. For vitrification, states in $\mathfrak{S}_{\mathbf{X}}$ refer to the
(physical or hypothetical) EQ states of the supercooled liquid. Defining such
as restricted form of the equilibrium state space is very common in
theoretical physics. For example, when we talk about an equilibrium crystal of
a material, it is also defined in a restricted sense in which its molecules
are not supposed to dissociate into constituent atoms. From now on, we will
denote EQ quantities either by a subscript "eq" or "SCL" and NEQ quantities
without any subscript. If we are interested in a ordered state, we will use a
subscript "CR" to denote its quantity.

\subsection{Equilibrium State}

In EQ thermodynamics, a body is specified by a set $\mathbf{X}$\ formed by its
independent extensive observables ($E,V,N$, etc.); the set also serves the
purpose of specifying the thermodynamic \emph{state} (also known as the
\emph{macrostate}) $\mathsf{M}$\ of the system. All EQ states belong to the
state space $\mathfrak{S}_{\mathbf{X}}$ as said above. The thermodynamic
entropy of the body in equilibrium is a \emph{state function} of $\mathbf{X}$
and is written as $S_{\text{eq}}(\mathbf{X})$. It is one of the state
functions of the system and is supposed to be differentiable except possibly
at phase transitions, which we will not consider in this review. It satisfies
the Gibbs fundamental relation%
\begin{equation}
dS_{\text{eq}}(\mathbf{X})=(dE+P_{0}dV-\mu_{0}dN+...)/T_{0}%
,\label{Gibbs-Equilibrium}%
\end{equation}
where we have shown only the terms related to $E,V$ and $N$. The missing terms
refer to the remaining variables in $\mathbf{X\equiv}\left\{  X_{p}\right\}
$, and $T_{0},P_{0},\mu_{0},$ etc. have their standard meaning in equilibrium%
\begin{equation}
\frac{\partial S_{\text{eq}}}{\partial E}\doteq\frac{1}{T_{0}},\frac{\partial
S_{\text{eq}}}{\partial V}\doteq\frac{P_{0}}{T_{0}},\frac{\partial
S_{\text{eq}}}{\partial N}\doteq-\frac{\mu_{0}}{T_{0}},\cdots
.\label{Fields-Eq}%
\end{equation}
We have used a subscript $0$ since in equilibrium, the fields of $\Sigma$ and
$\widetilde{\Sigma}$ are the same.

\subsection{Nonequilibrium States and Internal Equilibrium
States\label{Sec-InternalEquilibrium}}

The above conclusion is most certainly not valid for a body out of
equilibrium. If the body is not in equilibrium with its medium, its
(macro)state $\mathsf{M}(t)$ will continuously change (relax), which is
reflected in the changes in all of its physical quantities $q(t)$ with time.
Such variations mean that the states no longer belong to $\mathfrak{S}%
_{\mathbf{X}}$. These states belong to the \emph{enlarged} state space
$\mathfrak{S}_{\mathbf{Z}}$ spanned by $\mathbf{Z=(X,}\boldsymbol{\xi
}\mathbf{)}$. The set $\boldsymbol{\xi}$ of internal variables
\cite{deGroot,Prigogine,Maugin,Beris,Kestin,Woods,Coleman,Jou0,Meixner} cannot
be controlled from the outside \cite{Note1}; a readable history of internal
variables is available in a recent paper by Maugin \cite{Maugin-2}. They are
used to characterize internal structures or inhomogeneity
\cite{deGroot,Prigogine,Maugin,Coleman,Jou0,Gujrati-I,Gujrati-II,Gujrati-III,Langer,Langer1,Langer2}
in the system, and are independent of the observables in $\mathbf{X}$ away
from equilibrium but become dependent on $\mathbf{X}$ in equilibrium. From
Theorem 4 in \cite{Gujrati-II}, it follows that with a proper choice of the
number of internal variables, the entropy can be written as $S(\mathbf{Z}(t))$
with no explicit $t$-dependence. The situation is now almost identical to that
of a body in equilibrium:\ The entropy is a function of $\mathbf{Z}(t)$ with
no explicit time-dependence. This allows us to identify $\mathbf{Z}(t)$\ as
the set of NEQ \emph{state variables}. States for which the entropy $S$
becomes a \emph{state function} of the \emph{state variable} $\mathbf{Z}$ are
called \emph{internal equilibrium }(IEQ) states
\cite{Gujrati-I,Gujrati-II,Gujrati-III,Langer,Langer1,Langer2,Woods,Maugin}
and we write%
\[
S_{\text{ieq}}(t)=S(\mathbf{Z}(t))
\]
for their entropy. This allows us to extend Eq. (\ref{Gibbs-Equilibrium}) to%
\begin{equation}
dS_{\text{ieq}}(t)=%
{\textstyle\sum\nolimits_{p}}
\left(  \partial S_{\text{ieq}}(t)/\partial Z_{p}(t)\right)  dZ_{p}%
(t)\label{Gibbs_Fundamental_Extended}%
\end{equation}
in which the partial derivatives are related to the fields of the system:%
\begin{align}
\frac{\partial S_{\text{ieq}}(t)}{\partial E(t)}  & \doteq\frac{1}{T(t)}%
,\frac{\partial S_{\text{ieq}}(t)}{\partial V(t)}\doteq\frac{P(t)}%
{T(t)},\nonumber\\
\frac{\partial S_{\text{ieq}}(t)}{\partial N(t)}  & \doteq-\frac{\mu(t)}%
{T(t)},\cdots,\frac{\partial S_{\text{ieq}}(t)}{\partial\boldsymbol{\xi}%
(t)}\doteq\frac{\mathbf{A}(t)}{T(t)};\label{Fields_System}%
\end{align}
these fields will change in time unless the system has reached equilibrium. It
is customary to call $\mathbf{A}$ the \emph{affinity} \cite{Donder}. For a
fixed $\mathbf{Z}$, $S_{\text{ieq}}$\ does not change in time. Hence, it must
have the maximum possible value for fixed $\mathbf{Z}$
\cite{Gujrati-Entropy1,Gujrati-Entropy2}. The EQ value of $\mathbf{A}$
vanishes \cite{deGroot,Prigogine}:
\begin{equation}
\mathbf{A}_{\text{eq}}=0.\label{EQ-Affinity}%
\end{equation}
In this case, $S_{\text{ieq}}$\ is no longer a function of $\boldsymbol{\xi}$,
which means that $\boldsymbol{\xi}$ \ is no longer independent of $\mathbf{X}$.

We consider the extension of the derivation given earlier \cite{Gujrati-I} for
the entropy of $\Sigma_{0}$\ by including the internal variable contribution
to obtain as the statement of the second law:%
\begin{align}
\frac{dS_{0}(t)}{dt}  & =\left(  \frac{1}{T(t)}-\frac{1}{T_{0}}\right)
\frac{dE(t)}{dt}+\nonumber\\
& \left(  \frac{P(t)}{T(t)}-\frac{P_{0}}{T_{0}}\right)  \frac{dV(t)}{dt}%
+\frac{\mathbf{A}(t)}{T(t)}\cdot\frac{d\boldsymbol{\xi}(t)}{dt}%
\label{Total_Entropy_Rate}\\
& >0;\nonumber
\end{align}
for a NEQ state. As the entropy of an isolated system $\Sigma_{0}$ can only
increase, $dS_{0}(t)/dt$ cannot be negative, which explains the last
inequality above for a NEQ process. The strict inequality will be replaced by
an equality for $\Sigma_{0}$\ in equilibrium. Each term in the first equation
must be positive in accordance with the second law for a NEQ state.

It follows from Eq. (\ref{Second_Law}) that the above discussion also applies
to an interacting system in a medium for which $d_{\text{i}}S/dt$ is
nonnegative. Thus, we can apply it to a vitrification process in which the
energy decreases with time during \emph{isothermal} (fixed $T_{0}$ of
$\Sigma_{0}$) relaxation. We must, therefore, have%
\begin{equation}
T(t)>T_{0} \label{Temperature_Behavior}%
\end{equation}
during any relaxation (at a fixed temperature and pressure of the medium) so
that $T(t)$ approaches $T_{0}$ from above [$T(t)$ $\rightarrow$ $T_{0}^{+}$]
and becomes equal to $T_{0}$ as the relaxation ceases and the equilibrium is
achieved; the plus symbol is to indicate that the $T(t)$ reaches $T_{0}$ from above.

The relaxation times for different internal variables in $\boldsymbol{\xi}$
depend on their nature and do not have to be the same. Indeed, the spectrum of
relaxation times in various contexts such as in crystalline solids
\cite{Nowik} and glasses \cite{Goldstein} is intimately related to the
existence of internal variables. Therefore, the spectrum of relaxation times
will be pivotal in our discussion and will be picked up again in Sec.
\ref{Sec-Hierarchy}.

By attempting to describe NEQ properties of a system by invoking internal
variables, one is able to explain a broad spectrum of NEQ phenomena, but it
should be stated here that the choice and the number of state variables
included in $\mathbf{X}$ or $\mathbf{Z}$\ is not so trivial and must be
determined by the nature of the experiments \cite{Maugin}. As we will see in
Sec. \ref{Sec-Hierarchy}, the observation time $\tau_{\text{obs}}$ plays a
central role in determining the relevant state variables during an experiment:

\begin{proposition}
\label{Prop-RelevantStateVariables} The state variables\ that determine the
generalized NEQ\ thermodynamics are those whose relaxation times are longer
than $\tau_{\text{obs}}$.
\end{proposition}

\begin{proof}
The proposition will be justified in Sec. \ref{Sec-Hierarchy}; see the
paragraph containing Eq. (\ref{ObservationTime-Hierarchy}).
\end{proof}

We will assume here that $\mathbf{Z}$ has been specified.\emph{ }For any IEQ
states \textsf{M}$_{\text{ieq}}(\mathbf{Z})$, we have $\tau_{\text{ieq}%
}\lesssim\tau_{\text{obs}}<\tau_{\text{eq}}$, where we have introduced the
\emph{internal equilibration time }$\tau_{\text{ieq}}$ required for the system
to come to an IEQ state in $\mathfrak{S}_{\mathbf{Z}}$. As expected,
$\tau_{\text{ieq}}=\tau_{\text{ieq}}(\mathbf{Z})$ depends on $\mathbf{Z}$ but
we will not explicitly exhibit its state dependence unless clarity is needed.
These states appear for $\tau_{\text{obs}}\geq\tau_{\text{ieq}}$. There are
many other states in $\mathfrak{S}_{\mathbf{Z}}$ having nonstate entropies
that\ appear for $\tau_{\text{obs}}<\tau_{\text{ieq}}$. As $\tau_{\text{obs}%
}\rightarrow\tau_{\text{ieq}}$, we obtain an IEQ state \textsf{M}%
$_{\text{ieq}}(\mathbf{Z})$. Therefore, there appears a delicate balance
between $\tau_{\text{obs}}$ and what internal variables we can describe by our
thermodynamic approach using the concept of IEQ states. This leads us to
consider the hierarchy of relaxation times, which is taken in Sec.
\ref{Sec-Hierarchy}.

It may appear to a reader that the concept of entropy being a state function
is very restrictive. This is not the case as this concept, although not
recognized by several workers, is implicit in the literature where the
relationship of the thermodynamic entropy with state variables is
investigated. To appreciate this, we observe that the entropy of a body in
internal equilibrium \cite{Gujrati-I,Gujrati-II} is given by the Boltzmann
formula%
\begin{equation}
S(\mathbf{Z}(t))=\ln W(\mathbf{Z}(t)), \label{Boltzmann_S_Extended}%
\end{equation}
in terms of the number of microstates $W(\mathbf{Z}(t))$ corresponding to
$\mathbf{Z}(t)$. In classical NEQ thermodynamics \cite{deGroot}, the entropy
is always taken to be a state function. In the Edwards approach \cite{Edwards}
for granular materials, all microstates are equally probable as is required
for the above Boltzmann formula. Bouchbinder and Langer \cite{Langer} assume
that the NEQ entropy is given by Eq. (\ref{Boltzmann_S_Extended}). Lebowitz
\cite{Lebowitz} also takes the above formulation for his definition of the NEQ
entropy. As a matter of fact, we are not aware of any work dealing with
entropy computation that does not assume the NEQ entropy to be \ a state
function. This does not, of course, mean that all states of a system are IEQ
states. For states that are not in internal equilibrium, the entropy is not a
state function so that it will have an explicit time dependence. But, as shown
elsewhere \cite{Gujrati-II}, this can be avoided by enlarging the space of
internal variables. The choice of how many internal variables are needed will
depend on experimental time scales.

\section{Hierarchy among Relaxation Times and Enlarged State
Spaces\label{Sec-Hierarchy}}

We now classify state variables in a hierarchical manner as below. In IEQ
states, $\boldsymbol{\xi}$ has had enough time $\tau_{\text{obs}}%
=\tau_{\text{ieq}}<\tau_{\text{eq}}$ for \textsf{M}$_{\text{ieq}}$ to emerge.
But for $\tau_{\text{obs}}<\tau_{\text{ieq}}$, the states in $\mathfrak{S}%
_{\mathbf{Z}}$ have not had enough time for \textsf{M}$_{\text{ieq}}$ to
emerge so that their entropy is a nonstate function, which will continue to
increase if the system is left isolated until it reaches $S_{\text{ieq}%
}(\mathbf{X},\boldsymbol{\xi})$ and becomes a state function. The affinity
$\mathbf{A}$ corresponding to $\boldsymbol{\xi}$ is nonzero in \textsf{M}%
$_{\text{ieq}}$. If there were other internal variables $\boldsymbol{\xi
}^{\prime},\boldsymbol{\xi}^{\prime\prime},\boldsymbol{\xi}^{\prime
\prime\prime},\cdots$ in the system, with relaxation times $\tau^{\prime}%
,\tau^{\prime\prime},\tau^{\prime\prime\prime},\cdots$, respectively, that are
distinct from $\boldsymbol{\xi}$, then these must have equilibrated during
$\tau_{\text{ieq}}$ so that their affinities $\mathbf{A}^{\prime}%
,\mathbf{A}^{\prime\prime},\mathbf{A}^{\prime\prime\prime},\cdots$ have
vanished, implying that they are no longer independent of $\mathbf{X}$
($\mathbf{A}^{\prime}=\partial S/\partial\boldsymbol{\xi}^{\prime}=0$). This
means that the entropy does not depend on them. It is clear that
$\tau_{\text{ieq}}$ forms an upper bound for the relaxation times
$\tau^{\prime},\tau^{\prime\prime},\tau^{\prime\prime\prime},\cdots$. Thus,
they play no role in $\mathfrak{S}_{\mathbf{Z}}$. When the process is carried
out somewhat faster ($\tau_{\text{obs}}<\tau_{\text{eq}}$) than that required
for obtaining \textsf{M}$_{\text{eq}}(\mathbf{X})$, then $\boldsymbol{\xi}$
has not had enough time to "equilibrate" as we have discussed earlier
\cite{Gujrati-Entropy1,Gujrati-Entropy2} and $\mathbf{A}\neq0$.

Even if $S$ does not depend on $\boldsymbol{\xi}^{\prime},\boldsymbol{\xi
}^{\prime\prime},\boldsymbol{\xi}^{\prime\prime\prime},\cdots$, we will see in
Sec. \ref{Sec-Fictive} that they affect the thermodynamics of the system
indirectly, a fact that does not seem to have been appreciated. For the
moment, we will not consider the internal variables $\boldsymbol{\xi}^{\prime
},\boldsymbol{\xi}^{\prime\prime},\boldsymbol{\xi}^{\prime\prime\prime}%
,\cdots$. We will consider them later and will denote them collectively by
$\boldsymbol{\xi}_{\text{E}}=(\boldsymbol{\xi}^{\prime},\boldsymbol{\xi
}^{\prime\prime},\boldsymbol{\xi}^{\prime\prime\prime},\cdots)$.

The discussion below is somewhat abstract and intricate, and requires patience
on the part of the reader. The set-theoretic notation is perfectly suited for
the abstract nature of the discussion. Some readers may find the set-theoretic
notation cumbersome, but this is the price we must pay to make the discussion
comprehensive but compact.

To simplify our discussion, we \emph{assume} that all internal variables in
$\boldsymbol{\xi}$ are divided into \emph{nonoverlapping} groups
$\boldsymbol{\xi}_{n}$ indexed by $n=1,2,\cdots$. We further assume that all
internal variables in $\boldsymbol{\xi}_{n}$ have the same relaxation time
$\tau_{n}$ so that they equilibrate and are no longer independent of
$\mathbf{X}$ for $\Delta t\gtrsim\tau_{n}$. The relaxation times depend
strongly on $\mathbf{X}$. Let us also define $\boldsymbol{\xi}_{0}=\mathbf{X}$
in order to compactify our notation below. Because of this, we can include
$\boldsymbol{\xi}_{0}=\mathbf{X}$ whenever we speak of internal variables from
now on, unless clarity is needed. The groups $\boldsymbol{\xi}_{n}%
,n=0,1,2,\cdots$ are indexed by $n$ so that $\tau_{n}$'s appear in a
\emph{decreasing }order (with $\tau_{0}=\tau_{\text{eq}}$):
\begin{equation}
\tau_{0}>\tau_{1}>\tau_{2}>\cdots. \label{RelaxationTime-Hierarchy}%
\end{equation}
The relaxation times form a discrete set and not a continuum for simplicity.
It is important that the set $\left\{  \boldsymbol{\xi}_{k}\right\}  $ has a
finite though large number of elements for a physically sensible thermodynamic
description of the system; having an enormous number of elements will make the
description unnecessarily too complex and completely useless for thermodynamics.

We now introduce the sequence of state spaces $\left\{  \mathfrak{S}%
_{n}\right\}  $, where $\mathfrak{S}_{n},n=0,1,2,\cdots$ is spanned by the
\emph{union}%
\[
\boldsymbol{\xi}^{(n)}\doteq\cup_{k=1}^{n}\boldsymbol{\xi}_{k},n\geq1,
\]
of all $\boldsymbol{\xi}_{k},k\leq n$, with relaxation times $\tau_{k}%
>\tau_{n+1}$, with $\boldsymbol{\xi}^{(0)}$ (not to be confused with
$\boldsymbol{\xi}_{0}=\mathbf{X}$) denoting an \emph{empty} set, so that
\[
\mathbf{Z}_{n}\doteq(\mathbf{X},\boldsymbol{\xi}_{1},\boldsymbol{\xi}%
_{2},\cdots,\boldsymbol{\xi}_{n})\equiv(\boldsymbol{\xi}_{0},\boldsymbol{\xi
}^{(n)}),n\geq0.
\]
Thus, $\mathfrak{S}_{0}=\mathfrak{S}_{\mathbf{X}}$, formed by $\mathbf{Z}%
_{0}=\boldsymbol{\xi}_{0}=\mathbf{X}$, is relevant when $\tau_{0}%
>\tau_{\text{obs}}>\tau_{1}$. Similarly, $\mathfrak{S}_{1}$, formed by
$\mathbf{Z}_{1}=(\boldsymbol{\xi}_{0},\boldsymbol{\xi}^{(1)})=(\mathbf{X}%
,\boldsymbol{\xi}_{1}),$ is relevant when $\tau_{1}>\tau_{\text{obs}}>\tau
_{2}$, and so on.

It is clear from the construction that the state spaces $\mathfrak{S}%
_{n},n=0,1,2,\cdots$ are ordered with \emph{increasing} dimensions:
\begin{equation}
\mathfrak{S}_{0}\subset\mathfrak{S}_{1}\subset\mathfrak{S}_{2}\subset\cdots.
\label{StateSpace-Hierarchy}%
\end{equation}
The longest relaxation time in $\mathfrak{S}_{n}$ is $\tau_{0}$ corresponding
to $\boldsymbol{\xi}_{0}=\mathbf{X}$\ and the shortest relaxation time is
$\tau_{n}$ corresponding to $\boldsymbol{\xi}_{n}$.\ Any $\boldsymbol{\xi}%
_{k},k>n$ with relaxation time shorter than $\tau_{n}$ need not be considered
as it has already equilibrated and does not affect any state in $\mathfrak{S}%
_{n}$. We can summarize this conclusion as the following

\begin{proposition}
\label{Prop-AdditionalInternalVariables}The additional internal variable
$\boldsymbol{\xi}_{k}$ in $\mathfrak{S}_{k}$ relative to $\mathfrak{S}_{k-1}$
equilibrates and plays no role (i.e., is absent) in all smaller state spaces
$\mathfrak{S}_{l},l\leq k-1$ but participate in all state spaces
$\mathfrak{S}_{l}$ larger than $\mathfrak{S}_{k-1}$, i.e., $l\geq k$.
\end{proposition}

\begin{proof}
See the discussion above.
\end{proof}

%

\begin{figure}
[ptb]
\begin{center}
\includegraphics[
height=2.3315in,
width=3.5622in
]%
{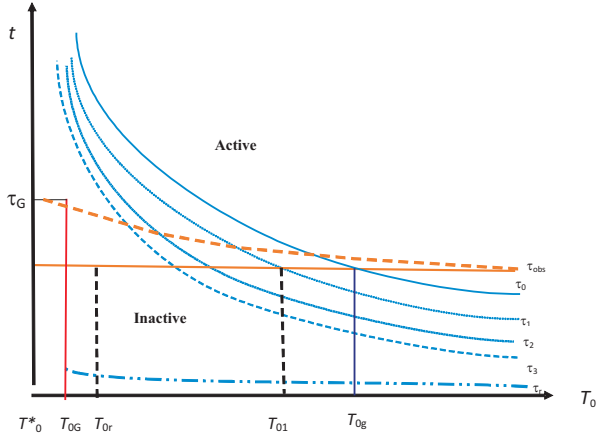}%
\caption{Schematic form of relaxation times $\left\{  \tau_{n}\right\}  $ as a
function of the temperature $T_{0}$ for a fixed pressure $P_{0}$ of the
medium. This figure will play an important role in the discussion of
vitrification later. At low enough temperatures near $T_{0}^{\ast
}<T_{0\text{G}}$, relaxation times become extremely large so that there is
practically no relaxation over a long period of time. However at
$T_{0}>T_{0\text{g}}$, all internal variables have equilibrated over
$\tau_{\text{obs}}$ in the figure. We have drawn $\tau_{\text{obs}}$ as a red
solid horizontal line when it does not change, and as a red broken line when
it increases, as $T_{0}$ is reduced.}%
\label{Fig-Relaxation}%
\end{center}
\end{figure}

Let us consider some observation time $\tau_{\text{obs}}$ used to observe a
state \textsf{M} of an interacting system. We can always find a pair of
\emph{neighboring} state spaces $\mathfrak{S}_{n+1}\supset\mathfrak{S}%
_{n},n\geq0$ satisfying
\begin{equation}
\tau_{n+1}<\tau_{\text{obs}}<\tau_{n}; \label{ObservationTime-Hierarchy}%
\end{equation}
the two sides define a \emph{window} $\Delta t_{n}\doteq\tau_{n}-\tau_{n+1}$
in which $\tau_{\text{obs}}$\ must lie. As $\tau_{\text{obs}}>\tau_{n+1}$%
,$\ $\ all $\boldsymbol{\xi}_{k}$'s$,k>n$, do not have to be considered to
describe the state \textsf{M} as they have already equilibrated (cf. the
discussion of $\boldsymbol{\grave{\xi}}$ above); thus, $\mathfrak{S}_{k},k>n$,
play no role in describing \textsf{M}. As $\tau_{\text{obs}}<\tau_{n}$, we
need to consider all $\boldsymbol{\xi}_{k},k\leq n$ to describe \textsf{M}. We
must, therefore, use $\mathfrak{S}_{n}$ to describe \textsf{M} for a
$\tau_{\text{obs}}$ in this window; we denote \textsf{M} by \textsf{M}%
$(\mathbf{Z}_{n})$ for clarity in this section. Among all the states in
$\mathfrak{S}_{n}$, there are IEQ states \textsf{M}$_{\text{ieq}}%
(\mathbf{Z}_{n})$ for which $S=S_{\text{ieq}}(\mathbf{Z}_{n})$. This happens
when $\tau_{\text{obs}}\simeq\tau_{\text{ieq}}(\mathbf{Z}_{n})\leq\tau_{n}$,
$\tau_{\text{ieq}}(\mathbf{Z}_{n})$\ denoting the time required for
\textsf{M}$(\mathbf{Z}_{n})$ to evolve into \textsf{M}$_{\text{ieq}%
}(\mathbf{Z}_{n})$; we will also use $\tau_{\text{ieq}}^{(n)}$ or simply
$\tau_{\text{ieq}}$ to denote $\tau_{\text{ieq}}(\mathbf{Z}_{n})$ in
$\mathfrak{S}_{n}$\ if no confusion will arise. For $n=0$, $\tau_{\text{ieq}%
}(\mathbf{Z}_{0})$ simply refers to $\tau_{\text{eq}}$.

There exists IEQ states \textsf{M}$_{\text{ieq}}(\mathbf{Z}_{n})$ in
$\mathfrak{S}_{n}$ for which $\boldsymbol{\xi}_{n}$ is no longer independent
of $\mathbf{X}$; for these states, $\tau_{\text{obs}}\simeq\tau_{\text{ieq}%
}(\mathbf{Z}_{n-1},\boldsymbol{\xi}_{n}(\mathbf{X}))\equiv\tau_{\text{ieq}%
}(\mathbf{Z}_{n-1})$. However, $\boldsymbol{\xi}_{n}\rightarrow\boldsymbol{\xi
}_{n}(\mathbf{X})$ as $t\rightarrow\tau_{n}$ even if \textsf{M}$(\mathbf{Z}%
_{n-1})$ in $\mathfrak{S}_{n-1}$ has not turned into \textsf{M}$_{\text{ieq}%
}(\mathbf{Z}_{n-1})$. As achieving internal equilibrium will take some
additional time, we have $\tau_{\text{ieq}}(\mathbf{Z}_{n-1})>\tau_{n}$. We
thus conclude that (with $\tau_{\text{ieq}}^{(0)}$ representing $\tau
_{\text{eq}}$)%
\begin{equation}
\tau_{\text{ieq}}^{(n)}<\tau_{\text{ieq}}^{(n-1)},n>0,
\label{IEQ-TimeScale-Inequality}%
\end{equation}
which will be assumed in this work.

We now consider the window%
\begin{equation}
\tau_{1}<\tau_{\text{obs}}<\tau_{0}. \label{ObservationTime-Hierarchy-10}%
\end{equation}
As $\tau_{0}>\tau_{\text{obs}}>\tau_{1}$, $\boldsymbol{\xi}_{1}$ has already
equilibrated so it need not be considered, but $\boldsymbol{\xi}%
_{0}=\mathbf{X}$ has not yet equilibrated. Thus, the entropy must be a
function only of the observables $\mathbf{X}$, which we must write as
$S_{\text{ieq}}(\mathbf{X}(t))$ as it continues to vary. As $\tau_{\text{obs}%
}\rightarrow\tau_{0}$, $\mathbf{X}(t)\rightarrow\mathbf{X}_{\text{eq}}$,
$S_{\text{ieq}}$ continues to increase until it finally reaches $S_{\text{eq}%
}$; there is no explicit time dependence as \emph{all} $\boldsymbol{\xi}_{k}%
$'s$,k>0$, have equilibrated; see also Landau and Lifshitz \cite{Landau} and
Wilks \cite{Wilks}, where NEQ states with respect to the medium are treated as
IEQ states in $\mathfrak{S}_{\mathbf{X}}$. This is the most common way NEQ
states in the literature are treated when internal variables are not invoked.
This is only possible when $\tau_{\text{obs}}$ satisfies Eq.
(\ref{ObservationTime-Hierarchy-10}).

We now consider the remaining case
\begin{equation}
\tau_{\text{obs}}\geq\tau_{0}. \label{ObservationTime-Hierarchy-0}%
\end{equation}
This situation corresponds to the quasistatic case so that even
$\boldsymbol{\xi}_{0}=\mathbf{X}$ has equilibrated to $\mathbf{X}_{\text{eq}}$
and we are dealing with an EQ state
\[
S=S_{\text{eq}}=S(\mathbf{X}_{\text{eq}}).
\]

We know that $\left\{  \tau_{n}\right\}  $ depend on the state of the system.
In vitrification that is of our primary interest here, they depend on the
temperature $T_{0}$. It is commonly believed that $\tau_{n}$'s increase with
decreasing $T_{0}$ as shown in Fig. \ref{Fig-Relaxation}, where we show them
as a function of $T_{0}$. From this figure, we observe that for a given
$\tau_{\text{obs}}$, drawn as a solid or broken line in red, $\boldsymbol{\xi
}_{\text{E}}$ correspond to the internal variables that lie in the
\emph{inactive} zone lying below $\tau_{\text{obs}}$. (Recall that
$\boldsymbol{\xi}_{0}=\mathbf{X}$ is now included in internal variables.) They
have all equilibrated. The \emph{active} zone corresponds to internal
variables that lie above $\tau_{\text{obs}}$. They have not equilibrated. For
higher temperatures ($T_{0}>T_{0\text{g}}$), all internal variables are
inactive. At lower temperature, some of them become active and make the system
out of equilibrium. At very low temperatures, all internal variables become
active for their NEQ role. We will discuss this figure further in Sec.
\ref{Sec-Vitrification}.

\section{General Consideration\label{Sec_General_Consideration}}

We have in Sec. \ref{Sec-Hierarchy} that for a given $\tau_{\text{obs}}$, we
can find the window $\Delta t_{n}$ satisfying Eq.
(\ref{ObservationTime-Hierarchy-10}), which then determines the state space
$\mathfrak{S}_{n}$ to describe any state \textsf{M} for the given
$\tau_{\text{obs}}$. The internal variables $\boldsymbol{\xi}_{k}$'s$,k>n$, do
not have to be considered as their affinities $\mathbf{A}_{k}$'s have vanished
for the given $\tau_{\text{obs}}$. However, the situation is somewhat
complicated for the following reason. As $\tau_{k}$'s are determined by
time-dependent $\mathbf{Z}_{n}$, the window will continue to change with time
for a given $\tau_{\text{obs}}$ so the value of $n$ will have to adjusted as
$\tau_{k}$'s change. The most simple solution for this complication is to
allow considering all the internal variables regardless of whether they have
equilibrated or not. The fact that $\mathbf{A}=0$ for equilibrated internal
variables means that their contribution to $d_{\text{i}}W$ will vanish so they
will not affect the Gibbs fundamental relation. Despite this, as we will see
later in Sec. \ref{Sec-Fictive}, these internal variables leave their mark in
relaxation. Therefore, from now on, we will consider the entire set
$\mathbf{Z}$ in the thermodynamic approach.

\subsection{First Law}

The infinitesimal heat exchange between the medium $\widetilde{\Sigma}$ and
the system $\Sigma$ will be denoted by $d_{\text{e}}Q(t)$; similarly, the
infinitesimal work done on $\Sigma$ by $\widetilde{\Sigma}$ will be denoted by
$d_{\text{e}}W(t)$. The subscript "e" is a reminder of the exchange. Then the
first law of thermodynamics is written as%
\begin{equation}
dE(t)\equiv d_{\text{e}}Q(t)-d_{\text{e}}W(t) \label{Standard_Heat_Sum}%
\end{equation}
in terms of \emph{exchange heat and work} $d_{\text{e}}Q(t)=T_{0}d_{\text{e}%
}S(t)$ and $d_{\text{e}}W(t)=P_{0}dV(t)$, respectively; see Sec.
\ref{Sec-Notation}. If there are other kinds of exchange work such as due to a
magnetic field, an exchange of particles, etc. they can be subsumed in
$d_{\text{e}}W(t)$. However, for simplicity, we will assume only the
pressure-volume work in this work. Both quantities are controlled from outside
the system. If the pressure $P(t)$ of the system is different from the
external pressure $P_{0}$ of the medium, then their difference gives rise to
the internal work $d_{\text{i}}^{V}W(t)\doteq\left(  P(t)-P_{0}\right)
dV(t)$, which is dissipated within the system; we have added a superscript as
a reminder that this particular internal work is due to volume variation. If
there are internal variables, they do not contribute to $d_{\text{e}}W(t)$ as
the corresponding EQ affinity $\mathbf{A}_{0}=0$. Despite this, the internal
variable $\boldsymbol{\xi}$ does internal work given by $d_{\text{i}%
}^{\boldsymbol{\xi}}W(t)\doteq\mathbf{A}(t)\mathbf{\cdot}d\boldsymbol{\xi}(t)$
and must be added to the internal work due to pressure difference. We thus
identify the internal work $d_{\text{i}}W(t)$ as%
\begin{equation}
d_{\text{i}}W(t)\doteq\left(  P(t)-P_{0}\right)  dV(t)+\mathbf{A}%
(t)\mathbf{\cdot}d\boldsymbol{\xi}(t), \label{Irreversible-Work}%
\end{equation}
and the net work is
\begin{equation}
dW(t)=d_{\text{e}}W(t)+d_{\text{i}}W(t)=P(t)dV(t)+\mathbf{A}(t)\mathbf{\cdot
}d\boldsymbol{\xi}(t), \label{GeneralizedWork}%
\end{equation}
a quantity that depends only on $\Sigma$ and is oblivious to the properties
of$\ \widetilde{\Sigma}$. Such a quantity is called a system-intrinsic (SI)
quantity. Introducing a new quantity \cite{Gujrati-II,Gujrati-Entropy2}%
\begin{equation}
d_{\text{i}}Q(t)\equiv d_{\text{i}}W(t),
\label{Irreversible_Heat_Work_equality}%
\end{equation}
and the net heat
\begin{equation}
dQ(t)\doteq d_{\text{e}}Q(t)+d_{\text{i}}Q(t), \label{GeneralizedHeat}%
\end{equation}
we can write the first law as%
\begin{equation}
dE(t)=dQ(t)-dW(t). \label{First_Law}%
\end{equation}
As $dE(t)$ and $dW(t)$ are both SI-quantities, $dQ(t)$ must also be a
SI-quantity. Thus, the above formulation of the first law is in terms of
quantities that refer to the system. There are no quantities that refer to
$\widetilde{\Sigma}$. We will call $dQ(t)$ and $dW(t)$ as the
\emph{generalized heat} $dQ(t)$ \emph{added to} and the \emph{generalized
work} $dW(t)$ \emph{done by} the system \cite{Gujrati-I,Gujrati-II}. We will
reserve \emph{exchange heat }and\emph{ work} for $d_{\text{e}}Q(t)=T_{0}%
d_{\text{e}}S(t)$ and $d_{\text{e}}W(t)=P_{0}dV(t)$, respectively, throughout
this work; see Sec. \ref{Sec-Notation}. Remembering this, we will also call
generalized heat and work as simply \emph{heat} and \emph{work}, respectively,
for brevity.

\subsection{Second Law}

The second law states that the irreversible (denoted by a suffix i) entropy
$d_{\text{i}}S$ generated in any infinitesimal physical process going on
within a system satisfies the inequality%
\begin{equation}
d_{\text{i}}S\geq0; \label{Second_Law_Inequality}%
\end{equation}
the equality occurs for a reversible process. For the isolated system
$\Sigma_{0}$, we must have, see Eq. (\ref{Isolated-system-irreversibility})
\begin{equation}
dS_{0}=d_{\text{i}}S_{0}=d_{\text{i}}S\geq0. \label{Second_Law}%
\end{equation}

As the thermodynamic entropy is not measurable except when the process is
reversible, the second law remains useless as a computational tool. In
particular, it says nothing about the rate at which the irreversible entropy
increases. Therefore, it is useful to obtain a computational formulation of
the entropy, the \emph{statistical entropy}. This will be done in the next
section. The onus is on us to demonstrate that the statistical entropy also
satisfies this law if it is to represent the thermodynamic entropy. This by
itself does not prove that the two are the same. It has not been possible to
show that the statistical entropy is identical to the thermodynamic entropy in
general. Here, we show their equivalence only when the NEQ thermodynamic
entropy is a \emph{state function} of NEQ state variables to be introduced below.

\subsection{Internal Equilibrium Thermodynamic}

\ For a body in internal equilibrium, its entropy $S$ is a function of $E,V$
and $\boldsymbol{\xi}$. Introducing the corresponding fields%
\begin{align}
\left(  \partial S/\partial E\right)    & =\beta(t)\doteq1/T(t),\left(
\partial S/\partial V\right)  =\beta(t)P(t),\nonumber\\
\left(  \partial S/\partial\mathbf{\xi}\right)    & =\beta(t)\mathbf{A}%
(t),\label{Fields_Body}%
\end{align}
we can write down the differential%
\[
dS(t)=\beta(t)[dE(t)+P(t)dV(t)+\mathbf{A}(t)\mathbf{\cdot}d\boldsymbol{\xi
}(t)],
\]
which can be inverted to express $dE(t)$ as follows:%
\begin{equation}
dE(t)=T(t)dS(t)-P(t)dV(t)-\mathbf{A}(t)\mathbf{\cdot}d\boldsymbol{\xi
}(t).\label{Gibbs_Fundamental_Equation}%
\end{equation}
Comparing with Eq. (\ref{First_Law}), we conclude an identity%
\begin{equation}
dQ(t)=T(t)dS(t),\label{Def-dQ}%
\end{equation}
regardless of the number of internal variables are used to describe $\Sigma$.

We now write $dQ=T_{0}d_{\text{e}}S(t)+T_{0}d_{\text{i}}S(t)+[T(t)-T_{0}%
]dS(t)=d_{\text{e}}Q(t)+T_{0}d_{\text{i}}S(t)+[T(t)-T_{0}]dS(t)$. From this
and using Eq. (\ref{Irreversible_Heat_Work_equality}), we conclude that%
\begin{align}
d_{\text{i}}Q(t)  & =T_{0}d_{\text{i}}S(t)+[T(t)-T_{0}]dS(t)\\
& =Td_{\text{i}}S(t)+[T(t)-T_{0}]d_{\text{e}}S(t)=d_{\text{i}}%
W(t),\label{diQ-diW}%
\end{align}
which can be used to express $d_{\text{i}}S(t)$ as follows
\begin{align}
T_{0}d_{\text{i}}S(t) &  =[T_{0}-T(t)]dS(t)+[P(t)-P_{0}]dV(t)\nonumber\\
&  +\mathbf{A}(t)\cdot d\boldsymbol{\xi}%
(t);\label{Irreversible-contributions1}\\
Td_{\text{i}}S(t) &  =[T_{0}-T(t)]d_{\text{e}}S(t)+[P(t)-P_{0}%
]dV(t)\nonumber\\
&  +\mathbf{A}(t)\cdot d\boldsymbol{\xi}%
(t).\label{Irreversible-contributions2}%
\end{align}
Since $dS(t),d_{\text{e}}S(t),dV(t)$ and $d\boldsymbol{\xi}(t)$ are
independent variations, each of the three contributions on the right side in
each equation must be non-negative
\begin{subequations}
\label{Irreversible-entropy-contributions}%
\begin{align}
\lbrack T_{0}-T(t)]dS(t) &  \geq0,\label{Irreversible-entropy-contributions1}%
\\
\lbrack T_{0}-T(t)]d_{\text{e}}S(t) &  \geq
0,\label{Irreversible-entropy-contributions2}\\
\lbrack P(t)-P_{0}]dV(t) &  \geq0,\label{Irreversible-entropy-contributions3}%
\\
\mathbf{A}(t)\cdot d\boldsymbol{\xi}(t) &  \geq
0,\label{Irreversible-entropy-contributions4}%
\end{align}
to comply with the second law requirement $d_{\text{i}}S(t)\geq0$; we are
assuming $T_{0}$ and $T(t)$ are positive. The factors $T_{0}-T(t),$
$P(t)-P_{0}$ and $\mathbf{A}(t)$\ in front of the extensive variations are the
corresponding \emph{thermodynamic forces} that act to bring the system to
equilibrium. In the process, each force has its own irreversible entropy
generation \cite{Gujrati-II}. The last inequality implies that each
independent component $\xi_{k}\in\boldsymbol{\xi}$ must satisfy $A_{k}%
(t)d\xi_{k}(t)\geq0$. There will be no irreversible entropy generation and the
equalities occur when thermodynamic forces vanish, which is the situation for
a reversible process.

It should be noted that Eq. (\ref{Irreversible-entropy-contributions2}) simply
states that heat exchanges (flows) from hot to cold. To see this, we use the
equality $d_{\text{e}}S(t)=d_{\text{e}}Q(t)/T_{0}$ to rewrite the equation as
$[T_{0}-T(t)]d_{\text{e}}Q(t)\geq0$. If $T_{0}>T(t)$, heat is exchanged to the
system; if $T_{0}<T(t)$, heat is exchanged from the system.

It follows from the last two inequalities in Eq.
(\ref{Irreversible-entropy-contributions}) that
\end{subequations}
\begin{equation}
d_{\text{i}}W(t)\geq0. \label{IrreversibleWork}%
\end{equation}
This means that $d_{\text{i}}W(t)$ truly represents irreversibility or
dissipation within the system. We note that while each term in $d_{\text{i}%
}W(t)$ is non-negative, this is not so for $d_{\text{i}}Q(t)$ written in the
form
\begin{subequations}
\label{Irreversible_Heat_0}%
\begin{align}
d_{\text{i}}Q(t)  &  =T_{0}d_{\text{i}}S(t)+[T(t)-T_{0}%
]dS(t),\label{Irreversible_Heat_01}\\
&  =T(t)d_{\text{i}}S(t)+[T(t)-T_{0}]d_{\text{e}}S(t)
\label{Irreversible_Heat_02}%
\end{align}
in which the first term is non-negative, but the second term is non-positive.
This not only means that \emph{the physics of }$d_{\text{i}}Q(t)$\emph{ and
}$d_{\text{i}}S(t)$\emph{ is very different} but also that
\end{subequations}
\begin{equation}
d_{\text{i}}Q(t)\leq T_{0}d_{\text{i}}S(t)\,,dQ(t)\leq T_{0}dS(t);
\label{Heat-inequalities}%
\end{equation}
the equalities occur only for isothermal ($T=T_{0}$) or adiabatic ($dS=0$) processes.

Let us consider the Helmholtz free energy $H(S,V,\boldsymbol{\xi}%
\mathbf{,}P_{0})=E(S,V,\boldsymbol{\xi})+P_{0}V(t)$
\cite{Gujrati-I,Gujrati-II} in terms of the external pressure $P_{0}$ of the
medium. We can treat $HH(S,V,\boldsymbol{\xi}\mathbf{,}P_{0})$ as an
SI-quantity by treating $P_{0}$ as a parameter. It is easy to see that%
\begin{equation}
dH=TdS(t)-[P(t)-P_{0}]dV(t)-\mathbf{A}(t)\cdot d\boldsymbol{\xi}%
(t)+V(t)dP_{0}. \label{Enthalpy_variation0}%
\end{equation}
The above differential clearly shows that the enthalpy $H$ is a function of
$S,V,\boldsymbol{\xi}$\textbf{,} and $P_{0}$. Recall that for an EQ state,
$H(S,P_{0})$ is not a function of $V$ so it is a Legendre transform of
$E(S,V)$ with respect to $V(t)$. In other words, $\partial H/\partial V=0$.
What we see from above that, for a NEQ states, $H$ is not a Legendre transform
of $E$ with respect to $V$. This is clearly seen by evaluating%
\[
\partial H/\partial V=P(t)-P_{0}\neq0,
\]
as the pressure difference need not vanish in an irreversible process. Despite
this, $dH$ has no irreversible component as we easily find that%
\begin{equation}
dH=TdS(t)-d_{\text{i}}W(t)+V(t)dP_{0}=d_{\text{e}}Q+V(t)dP_{0},
\label{Enthalpy_variation}%
\end{equation}
regardless of the number and nature of the internal variables; we have used
here Eqs. (\ref{Def-dQ}) and (\ref{Irreversible_Heat_Work_equality}). Thus,
$dH$ only contains exchange quantities as both terms on the right side are
controllable from outside the system. As such, it does not have any
spontaneous or irreversible relaxation. For an isobaric process, $dP_{0}=0$ so
$dH$ reduces to
\begin{equation}
dH=d_{\text{e}}Q. \label{Enthalpy_variation-Isobaric}%
\end{equation}
The above equality, which is well known for a reversible process, remains
valid no matter how irreversible a process is. Thus, it must remain valid for
supercooled liquids and glasses. Observe that just as $d_{\text{i}}E=0$, see
Eq. (\ref{EnergyExchange}), so is $d_{\text{i}}H=0$, with $d_{\text{e}%
}H=dH=d_{\text{e}}Q$.

Let us now consider the Gibbs free energy $G(S,V,\boldsymbol{\xi}%
\mathbf{,}T_{0},P_{0})\doteq E(S,V,\boldsymbol{\xi})-T_{0}S(t)+P_{0}V(t)$
\cite{Gujrati-I,Gujrati-II} in terms of the external temperature $T_{0}$ and
pressure $P_{0}$ of the medium. As is the case with the enthalpy, the Gibbs
free energy is also not a Legendre transform of $E(S,V,\boldsymbol{\xi})$ with
respect to $S(t)$ and $V(t)$. We find that%
\begin{align}
dG  & =[T(t)-T_{0}]dS(t)-d_{\text{i}}W(t)-S(t)dT_{0}+V(t)dP_{0}\nonumber\\
& =-T_{0}d_{\text{i}}S(t)-S(t)dT_{0}+V(t)dP_{0},\label{Gibbs_Variation}%
\end{align}
in which the first term can be identified as $d_{\text{i}}G\doteq
-T_{0}d_{\text{i}}S(t)$ and the remainder as $d_{\text{e}}G\doteq
-S(t)dT_{0}+V(t)dP_{0}$. At fixed $T_{0}$ and $P_{0}$, we have
\[
dG=d_{\text{i}}G=-T_{0}d_{\text{i}}S(t)\leq0,
\]
showing that the Gibbs free energy decreases during spontaneous relaxation
such as on a glass.%
\begin{figure}
[ptb]
\begin{center}
\includegraphics[
trim=0.000000in 0.000000in 0.798896in 0.000000in,
height=1.932in,
width=3.8458in
]%
{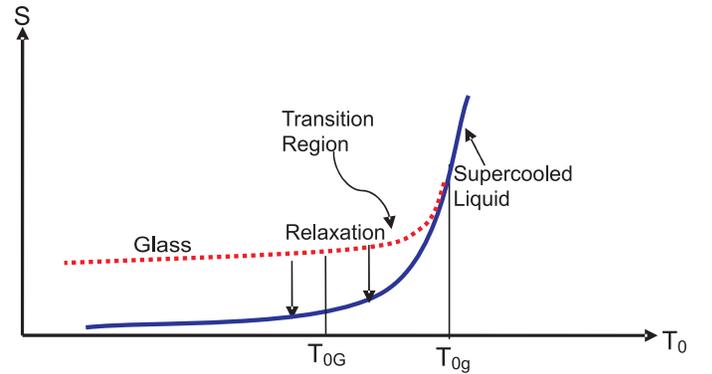}%
\caption{Schematic behavior of the entropy: equilibrated supercooled liquid
(solid curve) and a glass\ (dotted curve) during vitrification as a function
of the temperature $T_{0}$ of the medium. Structures appear to freeze at and
below $T_{0\text{G}}$; see text. The transition region between $T_{0\text{g}}$
and $T_{0\text{G}}$ over which the liquid turns into a glass has been
exaggerated to highlight the point that the glass transition is not a sharp
point.\ For all $T_{0}<T_{0\text{g}}$, the system undergoes isothermal (fixed
$T_{0}$) structural relaxation in time towards the supercooled liquid shown by
the downwards arrows. The entropy of the supercooled liquid is shown to
extrapolate to zero, but that of the glass to a positive value $S_{\text{R}}$
at absolute zero per our assumption. }%
\label{Fig_entropyglass}%
\end{center}
\end{figure}

\section{Entropy Bound during Vitrification\label{Sec-Vitrification}}

We now apply the IEQ thermodynamics of the last section to the vitrification
process, is carried out at some cooling rate as follows. The discussion in
this section is an elaboration and extension of our earlier discussion
\cite{Gujrati-Relaxation1,Gujrati-Relaxation2,Gujrati-Relaxation3,Gujrati-Relaxation4,Gujrati-Relaxation5,Gujrati-Entropy2}
and follows the approach first used by Bestul and Chang \cite{Bestul} and
later by Sethna and coworkers \cite{Sethna-Paper}. The temperature of the
medium is isobarically changed by some small but fixed $\Delta T_{0}$ from the
current value to the new value, and we wait for (not necessarily fixed) time
$\tau_{\text{obs}}$ at the new temperature to make an instantaneous
measurement on the system before changing the temperature again. At some
temperature $T_{0\text{g}}$, the relaxation time $\tau_{0}=\tau_{\text{eq}}$,
which continuously increases as the temperature is lowered [see Fig.
(\ref{Fig-Relaxation})], becomes equal to $\tau_{\text{obs}}$ as shown in Fig.
\ref{Fig_entropyglass}. The location of $T_{0\text{g}}$ depends on the rate of
cooling, i.e. on $\tau_{\text{obs}}$, which is clear from the figure. The
crossing $T_{0\text{g}}$ is lower for the broken $\tau_{\text{obs}}$\ than for
the solid $\tau_{\text{obs}}$. There are several other crossings at
$T_{01},T_{02},\cdots$, see Fig. (\ref{Fig-Relaxation}), at which
$\tau_{\text{obs}}$ crosses other relaxation curves for $\tau_{1},\tau
_{2},\cdots$, respectively. The crossing again depends on whether we take the
solid or the broken curve for $\tau_{\text{obs}}$. Let $T_{0\text{R}%
}>T_{0\text{G}}$ denote the temperature of the last such crossing (not shown
in the figure) before $T_{0\text{G}}$. Just below $T_{0\text{g}}$, the
structures are not yet frozen; they "freeze" at a lower temperature
$T_{0\text{G}}$ (not too far from $T_{0\text{g}})$ to form an amorphous solid
with a viscosity $\eta_{\text{G}}\simeq10^{13}$ poise corresponding to some
time scale $t_{\text{G}}$, see Fig. \ref{Fig_entropyglass}. This solid is
identified as a \emph{glass }determined by the choice of $\eta_{\text{G}}$ or
$t_{\text{G}}$. At $T_{0\text{G}}$, the relaxation time $\tau_{\text{R}}$ is
at least $\tau_{\text{G}}$.\ Over the glass transition region between
$T_{0\text{G}}$ and $T_{0\text{g}}$ in Fig. \ref{Fig_entropyglass}, the NEQ
liquid gradually turns from an EQ supercooled liquid at or above
$T_{0\text{g}}$ into a glass at or below $T_{0\text{G}}$, a picture already
known since Tammann \cite{Nemilov-Book}; see also \cite{Zhao}. Over this
region, some dynamical properties such as the viscosity vary continuously but
very rapidly. However, thermodynamic quantities such as the volume or the
enthalpy change continuously but slowly. As is evident from Fig.
(\ref{Fig-Relaxation}), more and more internal variables become active as the
temperature is reduced and will determine the thermodynamics in this region.
Below $T_{0\text{G}}$, all of these are almost "frozen" except those in the
inactive zone such as $\boldsymbol{\xi}_{\text{v}}$ corresponding to the
relaxation time $\tau_{\text{v}}$, representing local localized oscillations
within cells in the cell model \cite{Zallen}; see the discussion in Secs.
\ref{Sec-Fictive} and \ref{Sec-Conclusions}.

As the observation time $\tau_{\text{obs}}$ is increased, the equilibrated
supercooled liquid continues to lower temperatures before the appearance of
$T_{0\text{g}}$. In the \emph{hypothetical limit} $\tau_{\text{obs}%
}\rightarrow\infty$, it is believed that the equilibrated supercooled liquid
will continue to lower temperatures without any interruption, and is shown
schematically by the solid blue curve in Fig. \ref{Fig_entropyglass}. We
overlook the possibility of the supercooled liquid ending in a spinodal that
has been seen theoretically \cite{Gujrati-spinodal}. It is commonly believed
that this entropy will vanish at absolute zero ($S_{\text{SCL}}(0)\equiv0$),
as shown in the figure. As we are going to be interested in $S_{\text{SCL}%
}(T_{0})$ over $(0,T_{0\text{g}})$, we must also acknowledge the possibility
of an ideal glass transition in the system. If one believes in an ideal glass
transition, then there would be a singularity in $S_{\text{SCL}}(T_{0})$ at
some positive temperature $T_{\text{K}}<T_{0\text{G}},$ below which the system
will turn into an ideal glass whose entropy will also vanish at absolute zero
\cite[see also references cited there]{Gujrati-book}. The possibility of an
ideal glass transition, which has been discussed in a recent review elsewhere
\cite{Gujrati-book}, will not be discussed further in this work. All that will
be relevant in our discussion here is the fact that the entropy vanishes in
both situations ($S_{\text{SCL}}(0)\equiv0$). However, it should be emphasized
that the actual value of $S_{\text{SCL}}(0)$ has no relevance for the theorems
we derive below.

It is a common practice to think of the glass transition to occur at a point
that lies between $T_{0\text{g}}$ and $T_{0\text{G}}$.\ We have drawn entropy
curves (Glass and SCL) in Fig. \ref{Fig_entropyglass} for a process of
vitrification in a cooling experiment. The entropy curves $S_{\text{g}}%
(T_{0},t)\ $for Glass emerges out of $S_{\text{SCL}}(T_{0})$ at $T_{0\text{g}%
}$ for a given $\tau_{\text{obs}}$ in such a way that it lies above\ that of
SCL for $T_{0\text{g}}>T_{0}\geq0$. At any nonzero temperature $T_{0}$,
$S(T_{0},t)$ approach $S_{\text{SCL}}(T_{0})$ from above during
\emph{isothermal (fixed temperature of the medium) relaxation}; see the two
downward vertical arrows. These relaxations are discussed in the next section.

The concept of internal equilibrium is also a common practice now-a-days for
glasses \cite{Gutzow-Book,Nemilov-Book}. Employing the concept of internal
equilibrium provides us with an instantaneous Gibbs fundamental relation, see
Eq. (\ref{Gibbs_Fundamental_Equation}), which determines instantaneous
temperature, pressure, etc. of the system.

We now prove the entropy bounds%
\begin{equation}
S_{\text{R}}\equiv S(0)>S_{\text{expt}}(0)>S_{\text{SCL}}(0).
\label{ResidualEntropy_Bound}%
\end{equation}
in the form of Theorems \ref{Theorem_Lower_Bound} and
\ref{Theorem_Lower_Bound_SCL}. We will only consider isobaric cooling (we will
not explicitly exhibit the pressure in this section), which is the most
important situation for glasses. The process is carried out along some path
from an initial state A at temperature $T_{0\text{A}}$ in the supercooled
liquid state which is still higher than $T_{0\text{g}}$ to the state A$_{0}$
at absolute zero. The state A$_{0}$ depends on the path A$\rightarrow$A$_{0}$,
which is implicit in the following. The change $dS$ between two neighboring
points along such a path is
\cite{Donder,deGroot,Prigogine,Gujrati-I,Gujrati-II} $dS=d_{\text{e}%
}S+d_{\text{i}}S$; for a NEQ system, the two parts of $dS$ are path dependent.
The component
\begin{equation}
d_{\text{e}}S(t)=-d_{\text{e}}Q(t)/T_{0}\equiv C_{P}dT_{0}/T_{0}
\label{Heat_Capacity_Relation}%
\end{equation}
represents the reversible entropy exchange with the medium in terms of the
heat $d_{\text{e}}Q(t)$ given out by the glass at time $t$ to the medium whose
temperature at that instant is $T_{0}$. The component $d_{\text{i}}S>0$
represents the irreversible entropy generation in the irreversible process;
see Eq. (\ref{Second_Law_Inequality}). In general, it contains, in addition to
the contribution from the irreversible heat transfer with the medium,
contributions from all sorts of viscous dissipation going on \emph{within} the
system and normally require the use of internal variables
\cite{Donder,deGroot,Prigogine,Gujrati-I,Gujrati-II}. The equality in Eq.
(\ref{Second_Law_Inequality}) holds for a reversible process, which we will no
longer consider unless stated otherwise. The strict inequality $d_{\text{i}%
}S>0$ occurs only for an irreversible process such as in a glass.

\begin{theorem}
\label{Theorem_Lower_Bound}The experimentally observed (extrapolated) non-zero
entropy $S_{\text{expt}}(0)$ at absolute zero in a vitrification process is a
\emph{strict lower bound of the residual entropy} of any system:%
\[
S_{\text{R}}\equiv S(0)>S_{\text{expt}}(0).
\]

\end{theorem}

\begin{proof}
We have along A$\rightarrow$A$_{0}$%
\begin{equation}
S(0)=S(T_{0})+%
{\textstyle\int\limits_{\text{A}}^{\text{A}_{0}}}
d_{\text{e}}S+%
{\textstyle\int\limits_{\text{A}}^{\text{A}_{0}}}
d_{\text{i}}S,\label{General_Entropy_Calculation}%
\end{equation}
where we have assumed that there is no latent heat in the vitrification
process. The first integral is easily determined experimentally since it is
expressible in terms of the exchange heat%
\[%
{\textstyle\int\limits_{\text{A}}^{\text{A}_{0}}}
d_{\text{e}}S=-%
{\textstyle\int\limits_{\text{A}}^{\text{A}_{0}}}
\frac{d_{\text{e}}Q}{T_{0}}.
\]
The second integral in Eq. (\ref{General_Entropy_Calculation}) is always
\emph{positive}, but almost impossible to measure as it involves thermodynamic
forces; see Eq. (\ref{Irreversible-entropy-contributions1}):%
\begin{equation}%
{\textstyle\int\limits_{\text{A}}^{\text{A}_{0}}}
d_{\text{i}}S=\int\limits_{\text{A}}^{\text{A}_{0}}\frac{\left\{
\begin{array}
[c]{c}%
\lbrack T_{0}-T(t)]dS(t)+[P(t)-P_{0}]dV(t)\\
+\mathbf{A}(t)\cdot d\boldsymbol{\xi}(t)
\end{array}
\right\}  }{T_{0}}>0.\label{Irreversible Entropy loss}%
\end{equation}
It involves knowing and since the residual entropy $S_{\text{R}}$ is, by
definition, the entropy $S(0)$ at absolute zero, we obtain the important
result%
\begin{equation}
S_{\text{R}}\equiv S(0)>S_{\text{expt}}(0)\doteq S(T_{0\text{A}})+%
{\textstyle\int\limits_{T_{0\text{A}}}^{0}}
C_{P}dT_{0}/T_{0}.\label{Residual_Entropy_determination}%
\end{equation}
This proves Theorem \ref{Theorem_Lower_Bound}.
\end{proof}

The irreversibility during vitrification does not allow for the determination
of the entropy exactly, because evaluating the integral in Eq.
(\ref{Irreversible Entropy loss}) is not feasible
\cite{Gujrati-II,Nemilov-Book}. The forward inequality%
\[
S_{\text{R}}-S_{\text{expt}}(0)=%
{\textstyle\int\limits_{\text{A}}^{\text{A}_{0}}}
d_{\text{i}}S>0
\]
is due to the irreversible entropy generation from all possible sources
\cite{Donder,deGroot,Prigogine,Gujrati-I,Gujrati-II}. The inequality is made
strict as we are treating the NEQ glass with $\tau_{\text{obs}}<\tau
_{\text{eq}}(T_{0})$ and clearly establishes that the residual entropy at
absolute zero must be strictly larger than the "experimentally or
calorimetrically measured" $S_{\text{expt}}(0)$.

\begin{theorem}
\label{Theorem_Lower_Bound_SCL}The calorimetrically measured (extrapolated)
entropy during processes that occur when $\tau_{\text{obs}}<\tau_{\text{eq}%
}(T_{0})$ for any $T_{0}<T_{0\text{g}}$ is larger than the hypothetical
supercooled liquid entropy at absolutely zero
\[
S_{\text{expt}}(0)>S_{\text{SCL}}(0).
\]

\end{theorem}

\begin{proof}
Let $\dot{Q}_{\text{e}}(t)\equiv d_{\text{e}}Q(t)/dt$ be the rate of net heat
loss by the system during $\tau_{\text{obs}}<\tau_{\text{eq}}(T_{0})$ as it
relaxes isothermally at some fixed $T_{0}$. For each temperature interval
$dT_{0}<0$ below $T_{0\text{g}}$, we have
\begin{align*}
\left\vert d_{\text{e}}Q\right\vert  & \equiv C_{P}\left\vert dT_{0}%
\right\vert =%
{\textstyle\int\limits_{0}^{\tau_{\text{obs}}}}
\left\vert \dot{Q}_{\text{e}}\right\vert dt\\
& <\left\vert d_{\text{e}}Q\right\vert _{\text{eq}}(T_{0})\doteq%
{\textstyle\int\limits_{0}^{\tau_{\text{eq}}(T_{0})}}
\left\vert \dot{Q}_{\text{e}}\right\vert dt,\ \ \ \ \ \ T_{0}<T_{0\text{g}}%
\end{align*}
where $\left\vert d_{\text{e}}Q\right\vert _{\text{eq}}(T_{0})>0$ denotes the
net heat loss by the system to come to equilibrium, i.e. become supercooled
liquid during cooling at $T_{0}$. For $T_{0}\geq T_{0\text{g}}$, $dQ\equiv
dQ_{\text{eq}}(T_{0})\doteq C_{P\text{,eq}}dT_{0}$. Thus, the entropy loss
observed experimentally with $\tau_{\text{obs}}<\tau_{\text{eq}}(T_{0})$ is
less than the entropy loss if the system is allowed to come to SCL at each
temperature $T_{0}$. We thus conclude that
\begin{equation}
S_{\text{expt}}(0)>S_{\text{SCL}}(0).\label{Entropy_bound_at_0}%
\end{equation}

This proves Theorem \ref{Theorem_Lower_Bound_SCL}.
\end{proof}

The strict inequality above is the result of the fact that glass is a NEQ
state. Otherwise, we will have $S_{\text{expt}}(0)\geq S_{\text{SCL}}(0)$ for
any arbitrary state.

The difference $S_{\text{R}}-$ $S_{\text{expt}}(0)$ would be larger, more
irreversible the process is. The quantity $S_{\text{expt}}(0)$ can be
determined calorimetrically by performing a cooling experiment. We take
$T_{0\text{A}}$ to be the melting temperature $T_{0\text{M}}$, and uniquely
determine the entropy of the supercooled liquid at $T_{0\text{M}}$ by adding
the entropy of melting to the crystal entropy $S_{\text{CR}}(T_{0\text{M}})$
at $T_{0\text{M}}$. The latter is obtained in a unique manner by integration
along a reversible path from $T_{0}=0$ to $T_{0}=T_{0\text{M}}$:
\[
S_{\text{CR}}(T_{0\text{M}})=S_{\text{CR}}(0)+%
{\textstyle\int\limits_{0}^{T_{0\text{M}}}}
C_{P\text{,CR}}dT_{0}/T_{0},
\]
here, $S_{\text{CR}}(0)$ is the entropy of the crystal at absolute zero, which
is traditionally taken to be zero in accordance with the third law, and
$C_{P\text{,CR}}(T_{0})$ is the isobaric heat capacity of the crystal. This
then uniquely determines the entropy of the liquid to be used in the right
hand side in Eq. (\ref{Residual_Entropy_determination}). We will assume that
$S_{\text{CR}}(0)=0$. Thus, an experimental determination of $S_{\text{expt}%
}(0)$ is required to give the \emph{lower bound} to the residual entropy in
Eq. (\ref{ResidualEntropy_Bound}). Experimental evidence for a non-zero value
of $S_{\text{expt}}(0)$ is abundant as discussed by several authors
\cite{Giauque,Giauque-Gibson,Jackel,Gutzow-Schmelzer,Nemilov,Johari,Goldstein}%
; various textbooks \cite{Gutzow-Book,Nemilov-Book} also discuss this issue.
Goldstein \cite{Goldstein} gives a value of $S_{\text{R}}\simeq15.1$ J/K mol
for \textit{o-}terphenyl from the value of its entropy at $T_{0}=2$ K.
However, Eq. (\ref{Entropy_bound_at_0}) gives a mathematical justification of
$S_{\text{expt}}(0)>0$. The strict inequality proves immediately that the
residual entropy \emph{cannot} vanish for glasses, which justifies the curve
Glass in Fig. \ref{Fig_entropyglass}. The relevance of the residual entropy
has been discussed by several authors in the literature.
\cite{Pauling,Tolman,Sethna,Chow,Goldstein,Gujrati-Residual,Gujrati-Symmetry,Conradt,Gutzow-Schmelzer,Nemilov}%

By considering the state A$_{0}$ above to be a state A$_{0}$\ of the glass in
a medium at some arbitrary temperature $T_{0}^{\prime}$ below $T_{0\text{g}}$,
we can get a generalization of Eq. (\ref{Residual_Entropy_determination}):%
\begin{equation}
S(T_{0}^{\prime})>S_{\text{expt}}(T_{0}^{\prime})\doteq S(T_{0})+%
{\textstyle\int\limits_{T_{0}}^{T_{0}^{\prime}}}
C_{P}dT_{0}/T_{0}. \label{Entropy_determination}%
\end{equation}
We again wish to remind the reader that all quantities depend on the path
A$\rightarrow$A$_{0}$, which we have not exhibited. By replacing $T_{0}$\ by
the melting temperature $T_{0\text{M}}$ and $T_{0}^{\prime}$\ by $T_{0}$, and
adding the entropy $\widetilde{S}(T_{0\text{M}})$ of the medium on both sides
in the above inequality, and rearranging terms, we obtain (with $S_{\text{L}%
}(T_{0\text{M}})=S_{\text{SCL}}(T_{0\text{M}})$ for the liquid)%
\begin{equation}
S_{\text{L}}(T_{0\text{M}})+\widetilde{S}(T_{0\text{M}})\leq S(T_{0}%
)+\widetilde{S}(T_{0\text{M}})-%
{\textstyle\int\limits_{T_{0\text{M}}}^{T_{0}}}
C_{P}dT_{0}/T_{0}, \label{Setna_Inequality}%
\end{equation}
where we have also included the equality for a reversible process. This
provides us with an independent derivation of the inequality given by Sethna
and coworkers \cite{Sethna-Paper}.

It is also clear from the derivation of Eq. (\ref{Entropy_bound_at_0}) that
the inequality can be generalized to any temperature $T_{0}<T_{0\text{g}}$
with the result%
\begin{equation}
S_{\text{expt}}(T_{0})>S_{\text{SCL}}(T_{0}), \label{Entropy_bound}%
\end{equation}
with $S_{\text{expt}}(T_{0})\rightarrow S_{\text{SCL}}(T_{0})$ as
$T_{0}\rightarrow T_{0\text{g}}$ from below.

While we have only demonstrated the forward inequality, the excess
$S_{\text{R}}-S_{\text{expt}}(0)$ can be computed in NEQ thermodynamics
\cite{Donder,deGroot,Prigogine,Gujrati-I,Gujrati-II}, which provides a clear
prescription for calculating the irreversible entropy generation. The
calculation will, of course, be system-dependent and will require detailed
information. Gutzow and Scmelzer \cite{Gutzow-Schmelzer} provide such a
procedure with a single internal variable but under the assumption of equal
temperature and pressure for the glass and the medium. However, while they
comment that $d_{\text{i}}S\geq0$ whose evaluation requires system-dependent
properties, their main interest is to only show that it is negligible compared
to $d_{\text{e}}S$.

We have proved Theorems \ref{Theorem_Lower_Bound} and
\ref{Theorem_Lower_Bound_SCL}\ by considering only the system without paying
any attention to the medium. For Theorem \ref{Theorem_Lower_Bound}, we require
the second law, i.e. Eq. (\ref{Second_Law_Inequality}). This is also true of
Eq. (\ref{Entropy_determination}). The proof of Theorem
\ref{Theorem_Lower_Bound_SCL} requires the constraint $\tau_{\text{obs}}%
<\tau_{\text{eq}}(T_{0})$ for any $T_{0}<T_{0\text{g}}$, which leads to a NEQ
state. The same is also true of Eq. (\ref{Entropy_bound}).

We have focused on the system in this section. This does not mean that the
conclusion would be any different had we brought the medium into our
discussion. This is seen from the derivation of the inequality in Eq.
(\ref{Setna_Inequality}) from Eq. (\ref{Entropy_determination}).

\section{Entropy and Enthalpy during Isothermal
Relaxation\label{Sect_Relaxation}}

We wish to consider \emph{isothermal} relaxation in an isobaric cooling
experiment carried out at a fixed pressure $P_{0}$. Let us assume that
$\Sigma$ is in equilibrium at some temperature $T_{0}^{\prime}\leq
T_{0\text{g}}$ of some medium $\widetilde{\Sigma}^{\prime}$. We change to a
different medium $\widetilde{\Sigma}$ at $T_{0}<T_{0}^{\prime},P_{0}$ and
bring $\Sigma$ in its contact. Initially, the temperature $T\left(  0\right)
$ of $\Sigma$ is $T\left(  0\right)  =T_{0}^{\prime}>T_{0}$ so it is out of
equilibrium with the new medium and its temperature $T(t)$ will strive to get
closer to $T_{0}$ as we wait for $\Sigma$ to come to equilibrium with
$\widetilde{\Sigma}$; see Eq. (\ref{Temperature_Behavior}). The initial
entropy $S(T_{0},0)=S_{\text{SCL}}(T_{0}^{\prime})>S_{\text{SCL}}(T_{0})$. If
the system is now allowed to equilibrate, it will undergo spontaneous
(isothermal) relaxation at fixed $T_{0}$ so that $S(T_{0},t)\rightarrow
S_{\text{SCL}}(T_{0})$\ in time during which its temperature changes. We
assume that the relaxation times of $\boldsymbol{\xi}_{n}$ as a function of
$T(t)$ is similar to that shown in Fig. \ref{Fig-Relaxation}; all we need to
do is to replace $T_{0}$ by $T(t)$. During relaxation, the entropy of the
glass is supposed to decrease. This is what we expect intuitively as the
arrows show in Fig. \ref{Fig_entropyglass}. We now wish to consider such a
relaxation and determine the behavior of thermodynamic functions such as the
entropy, enthalpy, etc. using IEQ thermodynamics introduced above. We prove
two additional theorems in this section. The theorems are general even though
we have in mind NEQ states including glasses obtained under the condition
$\tau_{\text{obs}}<\tau_{\text{eq}}(T_{0})$ for any $T_{0}<T_{0\text{g}}$. We
consider the system to be in internal equilibrium with temperature $T(t)$,
pressure $P(t)$, etc. We remind the reader that all processes that go on
within the medium occur at constant temperature $T_{0}$, pressure $P_{0}$,
etc. Thus, there will not be any irreversible process going on within the
medium. All irreversible processes will go on within the system.

We will exploit below the \emph{strict} inequalities in Eq.
(\ref{Irreversible-entropy-contributions}) to derive a bound on the rate of
entropy variation. For a system out of equilibrium, the instantaneous entropy
$S(t)$ and volume $V(t)$ seem to play the role \cite{Gujrati-I} of "internal
variables," whose "affinities" are given by the corresponding thermodynamic
forces $T_{0}-T(t)$ and $P(t)-P_{0}$, respectively. This fact is not commonly
appreciated in the glass literature to the best of our knowledge. Even during
an isobaric vitrification, there is no fundamental reason to assume that the
pressure $P$ of the system is always equal to the external pressure $P_{0}$.
However, it is a common practice to assume the two to be the same, which may
not be a poor approximation in most cases. We will not generally make such an
approximation in this work.

We now state Theorem \ref{Theorem_Entropy_Variation}.

\begin{theorem}
\label{Theorem_Entropy_Variation}The entropy of a glass reaches that of the
supercooled liquid from above during relaxation at fixed $T_{0},P_{0}$ of the
mediums. Thus,%
\[
S>S_{\text{SCL}},
\]
so that the entropy variation in time has a unique direction as shown by the
\emph{downward arrows} in Fig. \ref{Fig_entropyglass}.
\end{theorem}

\begin{proof}
It follows from Eqs. (\ref{Temperature_Behavior}) and
(\ref{Irreversible-entropy-contributions1}) that for any NEQ state during
relaxation (fixed $T_{0},P_{0}$)%
\begin{equation}
dS(t)/dt<0; \label{Entropy_variation}%
\end{equation}
the inequality turns into an equality once equilibrium is reached. In other
words, during relaxation,%
\[
S(T_{0},P_{0},t)\rightarrow S_{\text{SCL}}^{+}(T_{0},P_{0});
\]
the plus symbol is again to indicate that the glass entropy reaches
$S_{\text{SCL}}(T_{0},P_{0})$ from above. This completes the proof of Theorem
\ref{Theorem_Entropy_Variation}.
\end{proof}

We have shown $T_{0},P_{0}$ in $S(T_{0},P_{0},t)\equiv S(T(t),P(t),A(t))$ to
emphasize that the result is general during any relaxation. In the derivation,
we have only used the second law. Being a general result, it should be valid
for any real glass. Above $T_{0\text{g}}$, the system is always in equilibrium
with the medium so its temperature is the same as $T_{0}$. Below
$T_{0\text{g}}$, when the system is not in equilibrium with the medium, then
$T(t)>T_{0}$ in accordance with Eq. (\ref{Temperature_Behavior}) based on the
experimental observation. Any theory, such as the one proposed in
\cite{Reiss,Gupta-JNCS,Kivelson,Gupta1} and known as the \emph{entropy loss
view} of the glass transition, in which $S(T_{0},P_{0},t)$ drops below
$S_{\text{SCL}}(T_{0},P_{0})$ so that
\begin{equation}
S(T_{0},t)\leq S_{\text{SCL}}(T_{0}). \label{EntropyLossView}%
\end{equation}
In this case, during relaxation, $dS(t)>0$ so that $(T_{0}-T(t))dS(t)<0$ in
direct conflict with Eq. (\ref{Irreversible-entropy-contributions1}), a
consequence of the second law. Such a theory then violates the second law as
first pointed out by Goldstein \cite{Goldstein}; we will revisit this issue in
the final section.

We now prove the following theorem:

\begin{theorem}
\label{Theorem_Enthalpy_Variation}For a glass, we must have $H(T_{0}%
,P_{0},t)>H_{\text{SCL}}(T_{0},P_{0})$ at all $T_{0}<$ $T_{0\text{g}}$, where
$S>S_{\text{SCL}}.$\ 
\end{theorem}

\begin{proof}
According to Eqs. (\ref{Temperature_Behavior}) and
(\ref{Irreversible-entropy-contributions2}), we conclude that $d_{\text{e}%
}Q=T_{0}$ $d_{\text{e}}S<0$ [cf. Eq. (\ref{Heat_Capacity_Relation})] while
relaxation is going on and vanishes as $T(t)\rightarrow T_{0}^{+}$. It then
follows from Eq. (\ref{Enthalpy_variation-Isobaric}) that
\begin{equation}
\frac{dH(t)}{dt}\leq0, \label{Total_Entropy_Rate_1}%
\end{equation}
a result that is consistent with experimental observations
\cite{Goldstein-Ann}. This completes the proof of the theorem.
\end{proof}

It follows from the behavior of the Gibbs free energy $G(t)=H(t)-T_{0}S(t)$
during relaxation ($dG(t)/dt\leq0$) that $dH\leq T_{0}$ $dS$, \textit{i.e.},
\begin{subequations}
\begin{equation}
\left\vert \frac{dH(t)}{dt}\right\vert \geq T_{0}\left\vert \frac{dS(t)}%
{dt}\right\vert \label{H-S-Bound0}%
\end{equation}
and%
\begin{align}
\left\vert \Delta H(T)\right\vert  & \doteq H(T_{0})-H_{\text{SCL}}%
(T_{0})\nonumber\\
& \geq T_{0}[S(T_{0})-S_{\text{SCL}}(T_{0})];T_{0}<T_{0\text{g}}%
.\label{H-S-Bound}%
\end{align}
The equality holds at $T_{0}=T_{0\text{g}}$. We can also obtain Eq.
(\ref{H-S-Bound0}) using $dH=T_{0}$ $d_{\text{e}}S\leq T_{0}$ $dS$.

From Eqs. (\ref{Enthalpy_variation0}) and
(\ref{Irreversible-entropy-contributions}), we also have
\end{subequations}
\begin{equation}
\left\vert \frac{dH(t)}{dt}\right\vert \geq T(t)\left\vert \frac{dS(t)}%
{dt}\right\vert . \label{H-S-Bound1}%
\end{equation}
The last bound is tighter than the bound in Eq. (\ref{H-S-Bound0}) and reduces
to the equality obtained earlier \cite{Gujrati-I} where $\boldsymbol{\xi}$ was
neglected. This equality there was used to infer Eq. (\ref{Entropy_variation}%
). We have just established that the conclusion remains unaltered even if we
consider internal variables.

In summary, the isothermal relaxation originates from the tendency of the
glass to come to thermal equilibrium during which its temperature $T(t)$
approaches $T_{0}$ from above in time. The relaxation process results in the
lowering of the corresponding Gibbs free energy in time, as expected due to
the second law. But it also results in the lowering of the corresponding
entropy as shown in Fig. \ref{Fig_entropyglass}, and the enthalpy during
vitrification; the latter is observed experimentally \cite{Goldstein-Ann}.

\section{Temperature Disparity due to Fast and Slow Variables:
Tool-Narayanaswamy Equation\label{Sec-Fictive}}

We have shown that for a given $\tau_{\text{obs}}$, we can partition
$\boldsymbol{\xi}$ into two distinct groups: one containing internal variable
$\boldsymbol{\xi}_{\text{E}}$ whose affinity has vanished and the other one,
which we now denote by $\boldsymbol{\xi}_{\text{N}}$ that has not equilibrated
and has a nonzero affinity $\mathbf{A}$. It is the active internal variables.
As $\boldsymbol{\xi}_{\text{E}}$ has equilibrated, its temperature, pressure,
etc. must be those of the medium, that is, $T_{0},P_{0}$, etc. It is the
inactive internal variable. On the other hand, the temperature, pressure, etc.
associated with different components of $\boldsymbol{\xi}_{\text{N}}$ must not
be those of the medium as there will be nonzero thermodynamic forces to bring
each to equilibrium in due course. This raises a very interesting question.
Because we are dealing with an IEQ state of the system, there is a
well-defined and unique thermodynamic definition of its temperature
$T(t)\doteq\partial E(t)/\partial S(t)$. This temperature also satisfies the
identity $dQ(t)=T(t)dS(t)$. How does $T(t)$ relate to \ temperatures of
$\boldsymbol{\xi}_{\text{E}}$ and $\boldsymbol{\xi}_{\text{N}}$? To make some
progress, we assume $\boldsymbol{\xi}_{\text{E}}$ and $\boldsymbol{\xi
}_{\text{N}}$ to be quasi-independent over $\tau_{\text{obs}}$. There is a
strong experimental evidence for this \cite{Debenedetti-Stillinger,Richert}.
However, there are observables in $\mathbf{X}$ that also participate in
relaxation. For example, $V$ will relax if $P\neq P_{0}$. Similarly, $E$ will
relax if $T\neq T_{0}$. As we have discussed earlier \cite{Gujrati-I}, one can
treat $E,V,$ etc. in $\mathbf{X}$ as internal variables with their affinities
$1/T-1/T_{0},P/T-P_{0}/T_{0}$, etc. that vanish once equilibrium is reached.
This is also seen from Eq. (\ref{Total_Entropy_Rate}), where the first two
terms have the same form as the last term involving $d\boldsymbol{\xi}$;
recall that $dS_{0}/dt=d_{\text{i}}S/dt$. Therefore, in this section, we will
continue to include $\boldsymbol{\xi}_{0}=\mathbf{X}$ in $\boldsymbol{\xi}$ as
we had done in Sec. \ref{Sec-Hierarchy}. This should not cause any confusion.
We only have to be careful to \emph{always} include $\boldsymbol{\xi}%
_{0}=\mathbf{X}$ to specify the system even when $\tau_{\text{obs}}>$
$\tau_{\text{eq}}=\tau_{0}$.

\subsection{A Black Box Model}

We consider a simple NEQ\ laboratory problem to model the above situation.
Consider a system as a \textquotedblleft black box\textquotedblright%
\ consisting of two parts at different temperatures $T_{1}$ and $T_{2}>T_{1}$,
but insulated from each other so that they cannot come to equilibrium. The two
parts are like slow and fast motions in a glass or $\boldsymbol{\xi}%
_{\text{E}}$ and $\boldsymbol{\xi}_{\text{N}}$, and the insulation allows us
to treat them as independent, having different temperatures. We assume that
there are no irreversible processes that go on within each part so that there
is no irreversible heat $d_{\text{i}}Q_{1}$ and $d_{\text{i}}Q_{2}$ generated
within each part. We wish to identify the temperature of the system, the black
box. To do so, we imagine that each part is added a certain
\emph{infinitesimal} amount of heat from outside, which we denote by
$dQ_{1}=d_{\text{e}}Q_{1}$ and $dQ_{2}=d_{\text{e}}Q_{2}$. We assume the
entropy changes to be $dS_{1}$ and $dS_{2}$. Then, we have for the net heat
and entropy change
\[
dQ=dQ_{1}+dQ_{2},dS=dS_{1}+dS_{2}.
\]
We introduce the temperature $T$ by $dQ=TdS$. This makes it a thermodynamic
temperature of the black box; see Eq. (\ref{Def-dQ}). Using $dQ_{1}%
=T_{1}dS_{1},dQ_{2}=T_{2}dS_{2}$, we immediately find

\qquad\qquad\qquad%
\[
dQ(1/T-1/T_{2})=dQ_{1}(1/T_{1}-1/T_{2}).
\]
By introducing $x=dQ_{1}/dQ$, which is determined by the setup, we find that
$T$ is given by%
\begin{equation}
\frac{1}{T}=\frac{x}{T_{1}}+\frac{1-x}{T_{2}}. \label{T_eff}%
\end{equation}
As $x$ is between $0$ and $1$, it is clear that $T$ lies between $T_{1}$ and
$T_{2}$ depending on the value of $x$. Thus, we see from this heuristic model
calculation that the thermodynamic temperature $T$ of the system is not the
same as the temperature of either parts, a common property of a system not in equilibrium.

If the insulation between the parts is not perfect, there is going to be some
energy transfer between the two parts, which would result in maximizing the
entropy of the system. As a consequence, their temperatures will eventually
become the same. During this period, $T$ will also change until all the three
temperatures become equal.

\subsection{Tool-Narayanaswamy Equation}

We turn to the general case of relaxation of thermodynamic properties. At high
enough temperatures, the time variation of $T(t)$ as it relaxes towards
$T_{0}$ can be described as a single simple exponential with a characteristic
time scale $\tau_{\text{eq}}$. This happens when all internal variables have
come to equilibrium during $\tau_{\text{obs}}>$ $\tau_{\text{eq}}$ so no
internal variables besides $\boldsymbol{\xi}_{0}$ are needed, a case discussed
by Landau and Lifshitz \cite{Landau} and by Wilks \cite{Wilks}.

At low temperatures, this is not true. There are \emph{quasi-independent} slow
and fast internal variables $\boldsymbol{\xi}_{\text{N}}$ and $\boldsymbol{\xi
}_{\text{E}}$ that are well known in glasses and supercooled liquids
\cite{Debenedetti-Stillinger,Richert}. The situation is similar to the black
box considered above. Both parts will strive to come to equilibrium with the
medium but they have widely separated relaxation times. As time goes on during
relaxation, some of the groups in $\left\{  \boldsymbol{\xi}_{\text{n}%
}\right\}  $ introduced in Sec. \ref{Sec-Hierarchy} becomes part of
$\boldsymbol{\xi}_{\text{E}}$ after equilibration as we have discussed there.
We first assume, for simplicity, that all active internal variables in
$\boldsymbol{\xi}_{\text{N}}$ have the same relaxation time $\tau_{1}$, i.e.,
they equilibrate together but have not equilibrated. The quasi-independence of
$\boldsymbol{\xi}_{\text{N}}$ and $\boldsymbol{\xi}_{\text{E}}$ immediately
leads to the following partition of the $S,E,V$ and $\boldsymbol{\xi}$ into
two contributions, one from each kind:
\begin{equation}
\mathbf{Z}(t)=\mathbf{Z}_{\text{E}}(t)+\mathbf{Z}_{\text{N}}(t).
\label{Partitions_S}%
\end{equation}
For example, quasi-independence gives the additivity $S(t)=S_{\text{E}%
}(t)+S_{\text{N}}(t)$, where $S_{\text{E}}(t)$\ and $S_{\text{N}}(t)$\ stand
for $S_{\text{E}}(E_{\text{E}}(t),V_{\text{E}}(t),\boldsymbol{\xi}_{\text{E}%
}(t))$ and $S_{\text{N}}(E_{\text{N}}(t),V_{\text{N}}(t),\boldsymbol{\xi
}_{\text{N}}(t))$, etc. Here, we have introduced $V_{\text{E}}(t)$ as the
volume difference $V-V_{\text{f}}\left(  t\right)  $ in terms of the free
volume $V_{\text{f}}\left(  t\right)  $ in the cell model in which
$V_{\text{f}}\left(  t\right)  $ allows for the molecules to move long
distances (liquid-like slow motion) over $\tau_{\text{obs}}$
\cite{Gujrati-GlassEncyclopedia}. Thus, $V_{\text{E}}$ corresponds to the fast
center of mass solid-like motion within the cells, which are in equilibrium
with the medium; see also Zallen \cite{Zallen}.

Let us now introduce the "energy fraction" $x(t)$ as%
\begin{equation}
x(t)\equiv dE_{\text{N}}(t)/dE(t),\ 1-x(t)\equiv dE_{\text{E}}%
(t)/dE(t),\ \label{x_definition}%
\end{equation}
at a given $t$, so that
\begin{align}
\partial S_{\text{N}}(t)/\partial E(t)  & =x(t)\partial S_{\text{N}%
}(t)/\partial E_{\text{N}}(t),\ \nonumber\\
\partial S_{\text{E}}(t)/\partial E(t)  & =[1-x(t)]\partial S_{\text{E}%
}(t)/\partial E_{\text{E}}(t).\label{x_T_derivatives}%
\end{align}

By definition, we have $\partial S_{\text{E}}(t)/\partial E_{\text{E}%
}(t)=1/T_{0}$, while $\boldsymbol{\xi}_{\text{N}}$ will have a temperature
different from this. Assuming internal equilibrium, we can introduce a new
temperature $T_{\text{N}}(t)$ by%
\begin{equation}
\partial S_{\text{N}}(t)/\partial E_{\text{N}}(t)=1/T_{\text{N}}(t).
\label{Fictive_Temp}%
\end{equation}
The following identity%
\begin{equation}
\frac{1}{T(t)}=\frac{1-x(t)}{T_{0}}+\frac{x(t)}{T_{\text{N}}(t)}
\label{Narayanaswamy_Decomposition}%
\end{equation}
easily follows from considering $\partial S(t)/\partial E(t)$ and using Eq.
(\ref{Partitions_S}) for $S(t)$ and Eq. (\ref{x_T_derivatives}). This equation
should be compared with (\ref{T_eff}) obtained above using a black box model
and is identical to the Tool-Narayanaswamy equation \cite{Goldstein-Ann} in
form, except that we have given thermodynamic definitions of $x(t)$ in
(\ref{x_definition}) and $T_{\text{N}}(t)$ in Eq. (\ref{Fictive_Temp}).

It is easy to extend the above calculation to the case of different groups
$\left\{  \boldsymbol{\xi}_{n}\right\}  $ belonging to $\boldsymbol{\xi
}_{\text{N}}$. The quasi-independence gives
\begin{equation}
\mathbf{Z}(t)=\mathbf{Z}_{\text{E}}(t)+%
{\textstyle\sum\nolimits_{n}}
\mathbf{Z}_{n}(t), \label{Partitions_S-general}%
\end{equation}
so that $S(t)=S_{\text{E}}(t)+%
{\textstyle\sum\nolimits_{n}}
S_{n}(t)$ with $S_{\text{E}}(E_{\text{E}}(t),V_{\text{E}}(t),\boldsymbol{\xi
}_{\text{E}}(t))$ and $S_{n}(E_{n}(t),V_{n}(t),\boldsymbol{\xi}_{\not n
}(t))$ as discussed above. For each $S_{n}$, we have its own temperature
$T_{n}$ using. It is now easy to see that Eq.
(\ref{Narayanaswamy_Decomposition}) is extended to
\begin{equation}
\frac{1}{T(t)}=\frac{1-x(t)}{T_{0}}+\sum_{n}\frac{x_{n}(t)}{T_{n}(t)},
\label{Narayanaswamy_Decomposition-General}%
\end{equation}
with $x_{n}(t)\equiv dE_{n}(t)/dE(t)$ and $1-x(t)=\sum_{n}x_{n}(t)$.

Let us now understand the significance of the above analysis. The partition in
Eqs. (\ref{Partitions_S}) and (\ref{Partitions_S-general}) along with the
fractions $x(t)$ and $x_{n}(t)$ shows that the partition satisfies a lever
rule: the relaxing glass can be \emph{conceptually} (but not physically)
thought of as a "mixture" consisting of different "parts" corresponding to
different temperatures and fractions. However, one of the temperatures is
$T_{0}$ of the medium, while $T_{n}(t)$'s denote the temperature of the parts
that are not equilibrated yet. As some of these parts equilibrate, their
temperature become $T_{0}$ and they add to the weight $1-x(t)$ for the
equilibrated internal variables. Thus, we see that while $\boldsymbol{\xi
}_{\text{E}}(t)$ may play no role in the IEQ thermodynamics, it still plays an
important role in relating the thermodynamic temperature $T(t)$ with those of
various groups of $\boldsymbol{\xi}(t)$. Thinking of a system conceptually as
a "mixture" of "parts" is quite common in theoretical physics. One common
example is that of a superfluid, which can be thought of as a "mixture" of a
normal viscous "component" and a superfluid "component". In reality, there
exist two simultaneous motions, one of which is "normal" and the other one is
"superfluid". A similar division can also be carried out in a superconductor:
the total current is a sum of a "normal current" and a "superconducting current".

Such an analysis has been carried out in detail earlier \cite{Gujrati-I},
where a connection is made with the notion of the "fictive" temperature
\cite{Goldstein-Ann} but in the absence of any internal variables (besides
$\boldsymbol{\xi}_{0}$). Here, we will summarize that discussion and refer the
reader to this work for missing details. It is easy to first consider the
simple case in Eq. (\ref{Narayanaswamy_Decomposition}). One can consider the
part $\boldsymbol{\xi}_{\text{N}}$ of the energy fraction $x(t)$\ at
$T_{\text{N}}$ to represent a "fictitious" SCL at temperature $T_{\text{N}}$.
It is fictitious since the entire system does not consist of this part so it
is not in equilibrium as SCL\ is supposed to be; it is missing the part
corresponding to the fraction $1-x(t)$. We can supplement mentally the
fictitous SCL by the same SCL of fraction $x(t)$ at the same temperature
$T_{\text{N}}$ to ensure that the entire system consists of $\boldsymbol{\xi
}_{\text{N}}$ at $T_{\text{N}}$. This now represents an IEQ state at
$T_{\text{N}}$, the left side of Eq. (\ref{Narayanaswamy_Decomposition}).
Thus, $T_{\text{N}}$ represents the thermodynamic temperature of this IEQ
state, which can then be treated as an "unequilibrated" SCL, in thermal
equilibrium with a medium at $T_{\text{N}}$ (but not at $T_{0}$). We have
identified it as an "unequilibrated" SCL since there is no reason for
$\mathbf{A}_{\text{N}}$ corresponding to $\boldsymbol{\xi}_{\text{N}}$ to
vanish in this SCL, whereas it is required to vanish in equilibrium. This
SCL\ at $T_{\text{N}}$\ is also not identical to the glass as the latter has
$\boldsymbol{\xi}_{\text{E}}(t)$ at $T_{0}$, which is absent in this SCL. We
can thus justify $\not T  _{\text{N}}$ as the fictive temperature.

This picture can be extended to Eq. (\ref{Narayanaswamy_Decomposition-General}%
) by introducing $T_{\text{N}}$ as follows:%
\[
\frac{x(t)}{T_{\text{N}}}=\sum_{n}\frac{x_{n}(t)}{T_{n}(t)},
\]
which converts it to Eq. (\ref{Narayanaswamy_Decomposition}). We can then
introduce an equilibrated SCL, in equilibrium with a medium at $T_{\text{N}}$
so that we can treat $T_{\text{N}}$ as the fictive temperature.

Instead of considering a derivative of $S$ with $E$, we can consider
derivatives with respect to other state variables such as $V$. In that case, a
similar analysis can be carried out as done in \cite{Gujrati-I} to obtain a
similar looking Tool-Narayanaswamy equation for $P(t)/T(t)$. We leave it to
the reader to carry out this simple extension. The result for $P=P_{0}$ is
given in \cite{Gujrati-I}.

\section{Discussions and Conclusions\label{Sec-Conclusions}}

\subsection{Consequence of the Relaxation Hierarchy}

We have presented a hierarchical classification of relaxation times in
increasing order in Eq. (\ref{RelaxationTime-Hierarchy}), which allows us to
determine a unique temporal window $\Delta t_{n}$ in Eq.
(\ref{ObservationTime-Hierarchy}) for a given $\tau_{\text{obs}}$ as shown by
the two neighboring relaxation curves around red the horizontal line at the
temperature $T_{0}$ of interest in Fig. \ref{Fig-Relaxation}. The discussion
is valid for any relaxing system with complex relaxation and is not restricted
to only SCL/glass undergoing vitrification. The temporal window is not fixed
as the state of the system changes so it must be adjusted appropriately; see
Fig. \ref{Fig-Relaxation}. Let us consider vitrification considered in Sec.
\ref{Sec-Vitrification}. Above $T_{0}=T_{0\text{g}}$, the system is always in
equilibrium (recall that we have used SCL as the equilibrium state) so
$\tau_{\text{obs}}\geq\tau_{0}$; see Eq. (\ref{ObservationTime-Hierarchy-0}).
There are no active internal variables. Therefore, the system's temperature
$T=T_{0}$. Slightly below $T_{0\text{g}}$ but above $T_{01}$, Eq.
(\ref{ObservationTime-Hierarchy-10}) is satisfied so $\boldsymbol{\xi}%
_{0}=\mathbf{X}$ is active, but all $\boldsymbol{\xi}_{k},k\geq1$ are inactive
so they need not be considered for a thermodynamic description. There are two
different contributions that affect the temporal window that needs to be considered:

\begin{enumerate}
\item[(i)] Cooling effect-As we lower $T_{0}$ from its previous value
$T_{0}^{\prime}=T_{0\text{g}}$ ($\tau_{\text{obs}}=\tau_{0}$), the system's
initial temperature is $T(0)=T_{0}^{\prime}$. As the system's temperature
determines $\tau_{1}$, it has the previous value $\tau_{1}^{\prime}$ at
$T_{0}^{\prime}$ initially so it lies below the curve $\tau_{1}$ at $T_{0}$.
But the value of $\tau_{0}$ at $T_{0}$ is determined by the new temperature
$T_{0}$ so it increases compared to $\tau_{0}^{\prime}=\tau_{0}(T_{0}^{\prime
})$. Consequently, we have $\tau_{\text{obs}}<\tau_{0}$ to satisfy Eq.
(\ref{ObservationTime-Hierarchy-10}).

\item[(ii)] Relaxation effect- During isothermal relaxation at the new
temperature, $T(t)$ decreases towards $T_{0}$, which increases $\tau_{1}$ from
$\tau_{1}^{\prime}=\tau_{1}(T_{0}^{\prime})$ to $\tau_{1}(T_{0})$. This
shrinks the window $\Delta t_{0}$ in Eq. (\ref{ObservationTime-Hierarchy-10})
in width to the width shown in Fig. \ref{Fig-Relaxation} at $T_{0}$.
\end{enumerate}

The discussion can be now applied to the sequence of cooling steps to $T_{0}$
between $T_{01}$ and $T_{02}$, between $T_{02}$ and $T_{03}$, etc. where we
are confronted with the new successive windows $\Delta t_{1},\Delta
t_{2},\cdots$. In each window, we need to consider newer internal variables
$\boldsymbol{\xi}_{1},\boldsymbol{\xi}_{2},\cdots$ so that in a window $\Delta
t_{n}$, we need to consider $\mathbf{Z}_{n}$\ consisting of $\mathbf{X}%
$\textbf{ }and $\boldsymbol{\xi}^{(n)}$ as discussed in Sec.
\ref{Sec-Hierarchy}. We thus conclude that the dimension of the state space
continues to grow during cooling until all internal variables (presumably
leaving out $\boldsymbol{\xi}_{\text{v}}$ that refers to local vibrations as
noted earlier) become active. Thus, in the glass transition region between
$T_{0\text{G}}$ and $T_{0\text{g}}$, the irreversibility continues to grow
until all internal variables become active.

\subsection{Residual Entropy}

As discussed above, we cannot just consider a fixed, small number of internal
variables (their number keeps changing in the transition region) if we want to
go to some small enough temperatures $T_{0}<T_{0\text{G}}$ and be able to
describe the cooling process thermodynamically. The best we can do is to
determine a large enough numbers of the internal variables that become active
in the transition region. This requires a deeper understanding of the
structure of glasses and identify these internal variables, which seems to be
an impossible task at present. In our view, this remains an unsolved problem
at present. Despite this, the inequalities in Eqs.
(\ref{ResidualEntropy_Bound}), (\ref{Entropy_bound_at_0}) and
(\ref{Entropy_bound}) remain valid for any choice of $\mathbf{Z}$.

As these inequalities are very important, we summarize them for the benefit of
the reader. According to Eq. (\ref{ResidualEntropy_Bound}), the residual
entropy $S_{\text{R}}$ cannot be less than the experimentally measured or
extrapolated $S_{\text{expt}}(0)$ at absolute zero; the latter itself cannot
be less than the entropy of the supercooled liquid at absolute zero. As we
have assumed $S_{\text{SCL}}(0)=0$, we claim the \emph{strict inequality}
\[
S_{\text{R}}>0.
\]
Indeed, the strict inequality between $S_{\text{expt}}(0)$ and $S_{\text{SCL}%
}(0)$ holds at all positive temperatures $T_{0}<T_{0\text{g}}$ as derived in
Eq. (\ref{Entropy_bound_at_0}).

We have not discussed the statistical formulation of the residual entropy,
which has been discussed by us in \cite[See Sec. 4.3.3]{Gujrati-Symmetry} and
\cite[See Sec. 7]{Gujrati-Entropy2}. The derivation does not require the use
of the second law or entropy maximization. Therefore, it applies to any
nonequilibrium state and is purely combinatorial in nature. For the sake of
completeness, we summarize the result. Let $\Gamma_{\lambda},\lambda
=1,2,\cdots,\mathcal{C}$ denote the number of disjoint components in the state
space, and let $p_{k_{\lambda}}$ denote the probability of a microstate
$k_{\lambda}$ in $\Gamma_{\lambda}$. The entropy $S=-%
{\textstyle\sum\nolimits_{\lambda}}
\sum_{k_{\lambda}}p_{k_{\lambda}}\ln p_{k_{\lambda}},$ $\sum_{\lambda}%
\sum_{k_{\lambda}}p_{k_{\lambda}}=1$, can be written as a sum of two parts:%
\[
S=%
{\textstyle\sum\nolimits_{\lambda}}
p_{\lambda}S_{\lambda}+S_{\mathcal{C}},
\]
where $p_{\lambda}\doteq\sum_{k_{\lambda}}p_{k_{\lambda}}$ is the probability
of the component $\Gamma_{\lambda}$, and $S_{\lambda}=-\sum_{k_{\lambda}%
}\widehat{p}_{k_{\lambda}}\ln\widehat{p}_{k_{\lambda}},\widehat{p}%
_{k_{\lambda}}\doteq p_{k_{\lambda}}/p_{\lambda}$, is the entropy of the
component, and
\begin{equation}
S_{\mathcal{C}}\doteq-%
{\textstyle\sum\nolimits_{\lambda}}
p_{\lambda}\ln p_{\lambda} \label{Component-ConfinementEntropy}%
\end{equation}
is the component confinement entropy. The residual entropy is the component
confinement entropy $S_{\mathcal{C}}$ at absolute zero, with $p_{\lambda
}=p_{\lambda0}$ denoting the probability of the component $\Gamma_{\lambda}$
at absolute zero. We have not imposed any \emph{equally probable assumption}
in the above derivation so the result is very general. However, to apply IEQ
thermodynamics, we need to impose equally probable assumption.

\subsection{Fate of the Entropy Loss Conjecture}

The isothermal relaxation considered in Sec. \ref{Sect_Relaxation} shows that
both $S$ and $H$ decrease with time, which is consistent with our intuitive
picture given at the start of that section for $S$, and experimental evidence
for $H$. As we have shown, the behavior is a consequence of the second law.
The entropy loss view (ELV) mentioned after Theorem
\ref{Theorem_Entropy_Variation} and proposed in
\cite{Reiss,Gupta-JNCS,Kivelson,Gupta1} results in the conclusion that
contradict our results. In particular, the view suggests that during
relaxation, the entropy increase since $S(T_{0},t)\leq S_{\text{SCL}}(T_{0})$;
see Eq. (\ref{EntropyLossView}). As Goldstein \cite{Goldstein} has shown, this
is a violation of the second law. These authors agree that in their view of
the glass transition, the glasses do violate the second law, while others
\cite{Jackel,Gutzow-Schmelzer,Nemilov,Johari,Conradt,Gujrati-Residual} argue
in favor of the second law. For most scientists, the fact that the entropy
loss view violates the second law should be a strong indication that the view
is unrealistic. But the debate persists as is evident from some recent reviews
\cite{Schmelzer,Gupta0,Mauro0,Takada0,Nemilov0}.

Here, we hope to settle the debate by pointing out a hitherto unrecognized
\emph{internal inconsistency} of the ELV, assuming its premise that the
glasses do violate the second law. In other words, the second law is not the
absolute truth of Nature. This means that all the inequalities in Eq.
(\ref{Irreversible-entropy-contributions}) must be reversed for the view to
hold. Since $dS>0$ in ELV during relaxation, it follows from the reverse
inequality in Eq. (\ref{Irreversible-entropy-contributions1}) that
$T(t)>T_{0}$, which is the same as Eq. (\ref{Temperature_Behavior}). From the
reverse inequality in Eq. (\ref{Irreversible-entropy-contributions2}), we
conclude that%
\begin{equation}
\lbrack T_{0}-T(t)]dH(t)<0. \label{ELV-Consequence0}%
\end{equation}
If we demand that the ELV follow the experimental evidence ($dH(t)/dt<0$), we
must conclude that $T_{0}>T(t)$, which contradicts the previous conclusion and
the ELV\ becomes internally inconsistent. If, however, we accept the previous
conclusion $T(t)>T_{0}$ to ensure that the ELV remain internally consistent,
then $dH(t)/dt>0$ in contradiction with experimental evidence. Thus, the mere
fact that the ELV satisfies the experimental evidence ($dH(t)/dt<0$) does not
mean that it is internally consistent in the entropy loss view. In other
words, demanding that the ELV is consistent with experiments disproves the ELV
conjecture. Even though we have considered the entropy loss view at different
times
\cite{Gujrati-Residual,Gujrati-Entropy2,Gujrati-Relaxation1,Gujrati-Relaxation2,Gujrati-Relaxation3,Gujrati-Relaxation4,Gujrati-Relaxation5}%
, we believe that the conclusion drawn above is the most direct demonstration
of the internal inconsistency of the ELV, despite the fact that we have
allowed it to contradict the second law.

\subsection{Significance of Inactive Internal variables}

Even though the IEQ thermodynamics only involves the active internal
variables, it is clear from Sec. \ref{Sec-Fictive} that even inactive internal
variables such as $\tau_{\text{r}}$ indirectly affect the thermodynamics
through the determination of the temperature, pressure, etc. of the system. In
retrospect, this is not so surprising once we recognize that the temperature
of the system is a thermodynamic quantity. However, the division of the
internal variables into active and inactive parts mean that the temperature of
the system must be different from the temperature $T_{0}$ during isothermal relaxation.

\begin{acknowledgments}
I wish to acknowledge useful comments from Gyan P. Johari and Sergei V.
Nemilov over the years.
\end{acknowledgments}


\begin{thebibliography}{99}
\bibitem {Note-IEQ}There are several interpretation of the IEQ state in the
literature such as by Simon \cite{Simon} and by Landau and Lifshitz
\cite{Landau}. There is some difference with our usage here. These authors do
not use any internal variables and restrict themselves to using the state
variable $\mathbf{X}$. However, Simon allows the use of IEQ for a part of the
degrees of freedom such as the vibrations that have come to equilibrium with
the medium, even though other degrees of freedom have not. He then speaks of
the third law being applicable to the entropy of the IEQ degrees of freedom
such as the vibrational entropy. We on the other hand use "equilibrated" for
vibrations (as emphasized by the inactive $\boldsymbol{\xi}_{\text{E}}$
introduced later) and "nonequilibrated" for others (as emphasized by the
active $\boldsymbol{\xi}_{\text{N}}$ introduced later) under this condition.
We reserve the use of the IEQ state for the entire system and do not make any
claim about the third law being applicable to all IEQ states, which may
include glasses.
\end{thebibliography}
\end{document}